 \documentclass[12pt]{article}
\usepackage{amsmath,natbib,amssymb,amsthm}
\usepackage{bm} \usepackage{bbm}
\usepackage[normalem]{ulem}
\usepackage{booktabs}
\usepackage{array,multirow}
\usepackage{mathtools}
\usepackage{geometry}
\usepackage[utf8]{inputenc}
\usepackage[T1]{fontenc}

\newcommand{\ra}[1]{\renewcommand{\arraystretch}{#1}}

\usepackage{xr}
\usepackage{titletoc}

\newcommand{\Exp}[2]{\mathbb{E}_{#2}\!\left[#1\right]}

\renewcommand{\Pr}[1]{\mathbb{P}\!\left(#1\right)}

\newcommand{\indep}{\perp\!\!\!\perp}
\newcommand{\bbR}{\mathbb{R}} \newcommand{\cC}{\mathcal{C}}

\newcommand{\z}{{\mathbf{z}}}

  \newcommand{\g}{{\mathbf{g}}}

\newcommand{\X}{{\mathbf{X}}} \newcommand{\Y}{{\mathbf{Y}}}

\newcommand{\1}{{\mathbf{1}}}

\newcommand{\cP}{{\mathcal{P}}}

\newcommand{\var}{\text{var}}

\newcommand{\cN}{{\cal N}}

\newcommand{\pnorm}[1]{{|\!|}#1{|\!|}} 

\usepackage{verbatim} 
\usepackage{enumitem} 
\usepackage{appendix} \usepackage{pdfpages}
\usepackage[usenames,dvipsnames]{pstricks} \usepackage{epsfig}
\usepackage{pst-grad} 

\theoremstyle{remark}
\newtheorem{theorem}{Theorem}

\newtheorem{algorithm}{Algorithm}

\newtheorem{definition}{Definition}
\newtheorem{assumption}{Assumption}

\newtheorem{proposition}[theorem]{Proposition}
\newtheorem{remark}{Remark}

\usepackage{fancyhdr}
\newcommand{\blind}{1}

\begin{document}
\thispagestyle{empty}
\def\spacingset#1{\renewcommand{\baselinestretch}%
{#1}\small\normalsize} \spacingset{1}
 \if1\blind
 {
   \title{\bf Estimating Malaria Vaccine Efficacy in the Absence of a Gold Standard Case Definition: Mendelian Factorial Design}
   \author{Raiden B. Hasegawa\textsf{\footnote{\textit{Address for correspondence:} Department of
 Statistics, The Wharton School, University of Pennsylvania, Philadelphia, PA
 19104, (e-mail: \texttt{raiden@wharton.upenn.edu})}}
 \;and\; 
 Dylan S. Small\textsf{\footnote{Department of
 Statistics, The Wharton School, University of Pennsylvania, Philadelphia, PA
 19104, (e-mail: \texttt{dsmall@wharton.upenn.edu})}}}
 \maketitle
 } \fi

 \if0\blind
 {
   \bigskip
   \bigskip
   \bigskip
   \begin{center}
     {\LARGE\bf Estimating Malaria Vaccine Efficacy in the Absence of a Gold Standard Case Definition: Mendelian Factorial Design\par}
 \end{center}
   \medskip
 } \fi

\vspace{-0.5cm}
\begin{abstract}
Accurate estimates of malaria vaccine efficacy require a reliable definition of a malaria case. However, the symptoms of clinical malaria are unspecific, overlapping with other childhood illnesses. Additionally, children in endemic areas tolerate varying levels of parasitemia without symptoms. Together, this makes finding a gold-standard case definition challenging. We present a method to identify and estimate malaria vaccine efficacy that does not require an observable gold-standard case definition. Instead, we leverage genetic traits that are protective against malaria but not against other illnesses, e.g., the sickle cell trait, to identify vaccine efficacy in a randomized trial. Inspired by Mendelian randomization, we introduce Mendelian factorial design, a method that augments a randomized trial with genetic variation to produce a natural factorial experiment, which identifies vaccine efficacy under realistic assumptions. A robust, covariance adjusted estimation procedure is developed for estimating vaccine efficacy on the risk ratio and incidence ratio scales. Simulations suggest that our estimator has good performance whereas standard methods are systematically biased. We demonstrate that a combined estimator using both our proposed estimator and the standard approach yields significant improvements when the Mendelian factor is only weakly protective. 

\end{abstract}

\vfill

\newpage
\spacingset{1.45} 
\setcounter{page}{1}

\section{Introduction}
In 2017, there were an estimated 219 million cases of malaria of which 92\% occurred in Africa and an estimated 435,000 malaria related deaths of which 93\% occurred in Africa; nearly every case of malaria in Africa was cause by the parasite {\it Plasmodium falciparum} \citep{who2018}. Pregnant women and children under the age of 5 are the most vulnerable groups affected by malaria. To date, of more than 30 vaccines under development, the only vaccine to undergo a pivotal phase III trial is the pre-erythrocytic vaccine RTS,S which has shown to have limited efficacy ($30-50\%$ reduction in incidence rates; \cite{mahmoudi2017efficacy}). Consequently, the continued development of an efficacious P. {\it falciparum} malaria vaccine has the potential for substantial public health impacts. With so many potential vaccines in the development pipeline, a critically important statistical challenge is to develop methods for estimating vaccine efficacy.  To date, accurate estimation of vaccine efficacy against clinical outcomes attributable to P. {\it falciparum} malaria requires defining reliable case definitions, a task that is made difficult by the unspecific presentation of malaria in endemic areas. 

In the absence of a gold standard case definition, efficacy is usually assessed by choosing an inexact case definition which may falsely exclude cases attributable malaria and falsely include non-malaria cases. The established definition of clinical malaria in malaria prevention trials is the presence of a fever with a temperature $\ge37.5^o\text{C}$ and a P. {\it falciparum} parasite density above a certain threshold, e.g., 2500 or 5000 parasites per $\mu$l of blood \citep{TER_KUILE_2003, rts2011, Olotu_2013, bejon2013efficacy}. Case definitions of this form inevitably exclude some true cases and include some false cases due to heterogeneity in immunity and endemicity, fever killing effects, and parasite density measurement error.  With specificity $<1$, estimates of vaccine efficacy will be biased downward. It has been shown in simulations that such case classification errors have the potential to introduce substantial bias in many settings \citep{small2010evaluating}. 
Real trial data suggests that these challenges with the standard case definition often go unaddressed. In a multi-site pooled analysis of phase II RTS,S trial data using a fixed 2500 parasite per $\mu$l cutoff across all study sites, \cite{bejon2013efficacy} reports estimates of vaccine efficacy against malaria that were as high as $60\%$ in sites with low parasite prevalence and as low as $4\%$ in high parasite prevalence sites. The authors suggest a biological reason for this pattern: RTS,S prevents fevers in only a portion of mosquito bites and in areas where parasite prevalence is high, children are likely bitten more often by infected mosquitos. However this pattern is also consistent with bias due to the fixed case definition having lower specificity in higher prevalence areas where, due to improved immunity, children can carry higher parasite loads without fever. The standard case definition leaves much unanswered: is the heterogeneity in vaccine efficacy epidemiologically important or just an artifact of bias arising from an inexact case definition?

To overcome the challenges that accompany inexact case definitions, this paper introduces a new method for identifying malaria vaccine efficacy using natural genetic variation, which we call {\it Mendelian factorial design}.  Importantly, the method does not depend on an inexact case definition based on the parasite density, instead using all fevers (or deaths) with any level of parasitemia to estimate efficacy. To identify vaccine efficacy, the new method requires finding and recording genetic variants that provide specific protection against clinical malaria and operate through a different biological pathway than the vaccine being evaluated. A running example in this paper is the sickle cell trait, a hemoglobinopathy that has been shown to protect against malaria at the blood-stage (erythrocytic) of the infection and the RTS,S vaccine, which confers protection against malaria at the pre-erythrocytic, liver-stage of the infection. There are many other vaccine types, e.g., transmission blocking and erythrocytic vaccines, and genetic variants that might satisfy these identifying conditions \citep{ndila2018human}. We will discuss these occasionally throughout the paper and more thoroughly in the discussion.

In the following section, we introduce Mendelian factorial design informally in the context of the more familiar use of genetic variants as instrumental variables in {\it Mendelian randomization} studies, to which it has many parallels. In \S \ref{sec:robust}, we present a precise definition of vaccine efficacy and malaria-attributable fevers in a potential outcome framework and propose an identification strategy that holds under a few realistic assumptions. We then provide a simple $\sqrt{n}$-consistent and asymptotically normal covariate-adjusted estimator that is robust to model misspecification and assess its performance in a simulation study over a range of settings. We develop an improved ``bounded" estimator that combines the strengths of our estimator with that of the (biased) standard estimator. We demonstrate that it provides nearly uniform improvement over both estimators. In Appendix \ref{sec:ttffev} of the Supplemental Web Materials we extend the Mendelian factorial design to identify and estimate vaccine efficacy in time-to-first event studies.

\section{Mendelian Factorial Design: Parallels with Mendelian Randomization}\label{sec:MFD_intro}

When an observed association between a non-randomized exposure and an outcome may be confounded by an unobserved common cause, attributing the association to a causal effect may be misleading \citep{rosenbaum2002}. An instrumental variable (IV) is a covariate that is associated with the exposure but whose only association to the outcome is through a direct pathway to the exposure \citep{martens2006instrumental}. The IV takes the place of the physical randomization in a randomized trial, influencing only the assignment of subjects to exposure or control, setting up a natural experiment and providing an avenue for estimating the causal effect of the non-randomized exposure on the outcome \citep{rassen2009instrumental}. 
When genetic variants are used as instrumental variables, the corresponding collections of methods are often labeled as the {\it Mendelian randomization} (MR)  approach \citep{smith2008Mendelian}. For example, \cite{Kang_2013} proposed using the hemoglobin S variant (HbS) as an instrument to identify the causal effect of malaria on stunted growth. It is known that heterozygote carriers (HbAS, sickle cell trait) receive protection against clinical malaria whilst those without the variant (HbAA) do not. The first column of Table \ref{tab:parallelAssump} summarizes the assumptions required of HbAS status such that it can be used in a MR study to identify the causal effect of malaria on stunting. Assumption (1) requires that HbAS status in fact influences exposure; assumption (2) ensures that HbAS status can be treated ``as-if" it were randomized; and assumption (3) says that HbAS only effects stunted growth through its influence on exposure. Assumption (4) is required for point identification of the causal effect, although there are other alternative assumptions that can be used to point identify a causal effect such as monotonicity of the effect \citep{burgess2017review}.

Like MR, Mendelian factorial design (MFD) leverages genetic variation to identify a causal effect of interest -- in this instance, the efficacy of a vaccine, or roughly, the proportion reduction of disease-attributable outcomes in a population when the vaccine is applied \citep{lachenbruch1998sensitivity}. However, instead of addressing the bias that often accompanies a non-randomized exposure,  MFD attends to the bias in estimating treatment efficacy arising from an inexact case definition. The natural experiment arising from MR is used for causal inference in the absence of an actual randomized experiment, enabling the investigator to distinguish causal effects from unobserved confounding. The same natural experiment is employed in MFD to augment a two-arm randomized trial and create a simple $2\times 2$ factorial experiment. The factorial design allows the investigator to distinguish efficacy against disease-attributable outcomes even when the case definition incorrectly classifies a material number of non-disease outcomes as cases.

\begin{table}\caption{Parallel assumptions for Mendelian randomization and Mendelian factorial design.\label{tab:parallelAssump}  \vspace{0.5cm}}
  \centering 
  \begin{tabular}{rp{5cm}p{6cm}}
    Assumption & \textit{Mendelian Randomization} & \textbf{Mendelian Factorial Design} \vspace{.1cm}\\
    \hline \vspace{-.25cm}\\
    (1) & \textbullet\,HbAS associated with malaria & \textbullet\, HbAS has protective efficacy against malaria-attributable fever\vspace{.25cm} \\ \vspace{.25cm}
    (2) & \textbullet\, No unmeasured confounders associated with HbAS and stunted growth &  \textbullet\, HbAS ``as-if" randomized\\ \vspace{.25cm}
    (3) & \textbullet\, Only direct pathway from HbAS to stunted growth is through malaria & \textbullet\, Protection provided by HbAS is specific to malaria-attributable fever \\ 
    (4) & \textbullet\, HbAS does not modify the effect of malaria on stunting & \textbullet\,  No interaction effect between HbAS and vaccine (i.e., independent protective pathways)
  \end{tabular}
\end{table}

Although MR and MFD address different sources of bias in different settings, their designs share many parallel features. This is illustrated in Table \ref{tab:parallelAssump}, where each assumption of MR in the malaria stunting study corresponds to a closely related assumption that underlies a MFD study of vaccine efficacy against malaria-attributable fever. In more general terms, Assumption (1) says that the {\it Mendelian gene} or {\it Mendelian factor}, e.g., HbAS, is relevant. In the MR study this means that it is associated with the non-randomized exposure and in the MFD study this implies that the Mendelian gene is protective against the disease-attributable outcome. Assumption (2) requires that the Mendelian gene be unconfounded with the outcome of interest or be ``as-if" randomized, which implies the former. For the MR study, assumption (3) says that the Mendelian gene has no pleiotropic effects \citep{smith2008Mendelian}, that is, HbAS only effects stunted growth through its influence on malaria and not through its influence on another modifiable exposure or stunting itself. The parallel assumption for a MFD study is that the Mendelian factor protects against outcomes of any-cause only through reducing disease-attributable outcomes. Finally, assumption (4) says that the causal effect in the MR study does not vary over different levels of the Mendelian gene, or in other words, they do not interact. Similarly, the corresponding assumption in a MFD study is that the vaccine and Mendelian factor do not interact. In other words, the vaccine prevents the same proportion of disease-attributable outcomes at different levels of the Mendelian gene. For example, this assumption is plausible when a malaria vaccine and HbAS provide protection against malaria-attributable fevers through independent biological pathways.

\begin{figure}[ht]
\begin{center}
	\begin{tabular}{cc}
		{\it Mendelian Randomization} & {\bf Mendelian Factorial Design} \\
		\includegraphics[scale=0.75]{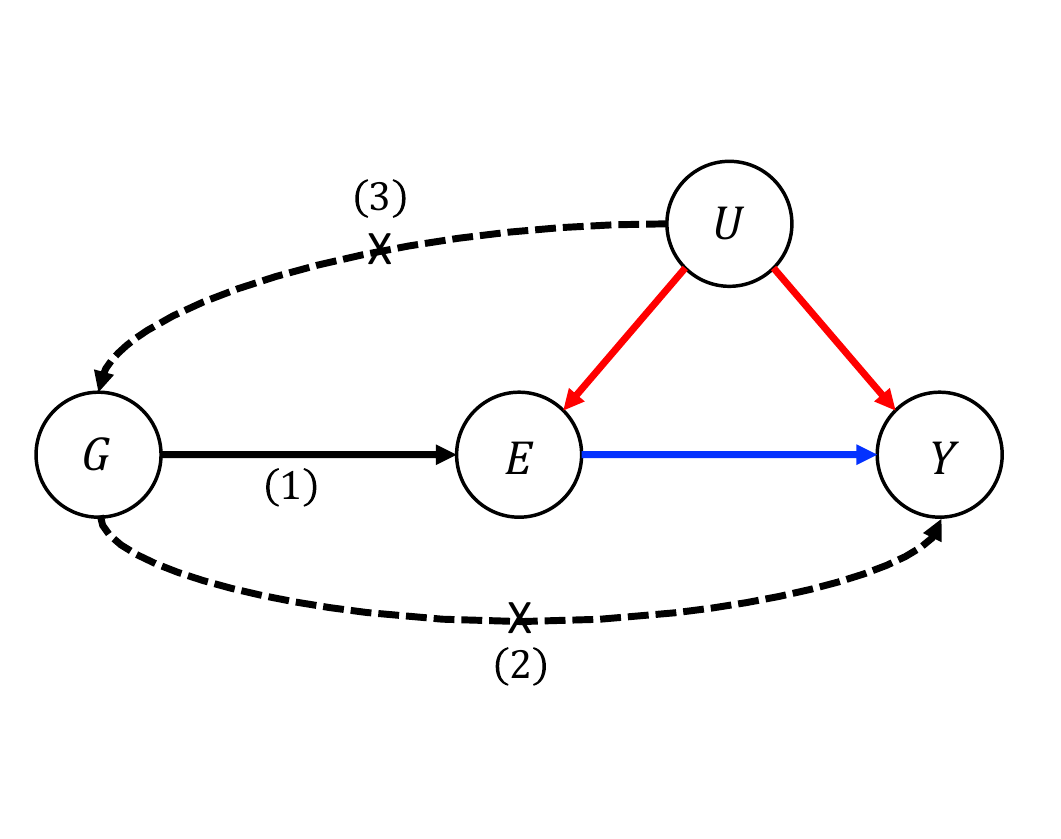} & \includegraphics[scale=0.75]{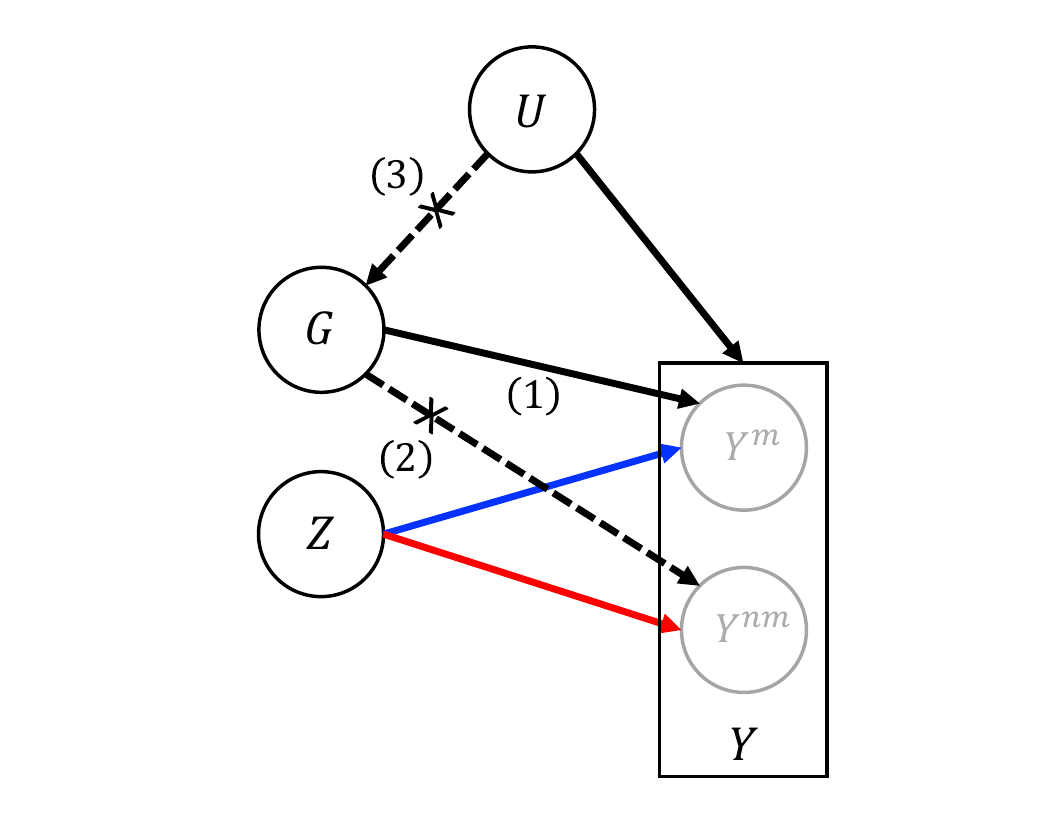}
	\end{tabular}
\caption[Causal diagrams for MR (left) and MFD (right).]{Causal diagrams for MR (left) and MFD (right). Blue arrows indicate the causal quantities of interest and red arrows indicate confounding addressed by each design. Gray variables $Y^m$ and $Y^{nm}$ are unobservable but $Y$, which doesn't distinguish between them is. Assumptions (1), (2) and (3) correspond to the similarly numbered assumptions in Table \ref{tab:parallelAssump}.}.
\label{fig:DAGs}
\end{center}
\end{figure}

Assumptions (1)-(3) in Table \ref{tab:parallelAssump} may be better understood encoded in a causal diagram. The causal diagram for MR and MFD are in the left and right panels of Figure \ref{fig:DAGs}, respectively. $G$ is a Mendelian gene (or factor), $U$ are unobserved confounders, $Z$ is a randomized vaccine, and $E$ is a non-randomized exposure. $Y$ is the outcome of interest -- in the MR study it is stunted growth and in the MFD study it is clinical malaria. $Y^m$ and $Y^{nm}$ are malaria-attributable and non-malaria fevers, respectively. We define these more precisely in \S \ref{sec:robust}. Finally, the blue arrows indicate the causal quantities of interest and the red arrows indicate the confounding addressed by each design. The causal diagram for MFD is a little bit unusual. $Y^m$ and $Y^{nm}$ are grayed out, emphasizing that we don't observed them, but instead only observe the outcome that is not distinguished by cause, $Y$. This is represented by the black bounding box. The MFD diagram hints at how the factorial structure and the absence of an arrow between $G$ and $Y^{nm}$ may be used to identify the arrow between $Z$ and $Y^m$, the vaccine efficacy.

\section{A Robust Framework for Estimating Vaccine Efficacy: Risk Ratios and Incidence Rate Ratios}\label{sec:robust}
\subsection{Notation: Observed Data}
Let $j = 1,\dots,J$ indicate sites in a multi-site randomized control trial (RCT) and $i=1,\dots,I_j$ indicate subjects at each center; then, $ij$ uniquely identifies each subject in the study. We let the total number of subjects in the trial be $n = \sum_j I_j$. For subject $ij$, let $Y_{ij}\in\{0,1\}$ be an observed fever or death with any parasitemia; let $Z_{ij} \in \{0,1\}$ indicate treatment/vaccine status and $G_{ij}\in\{0,1\}$ indicate sickle cell variant status ($G_{ij}=1$ if HbAS, $G_{ij}=0$ if HbAA); and let $X_{ij}\in \mathbb{R}^d$ be a $d$-dimensional vector of baseline characteristics. Because $G$ is assigned at conception, careful consideration of what variables constitute ``baseline" variables should be made. Let $D_{ij}$ and $U_{ij}$ indicate the malaria parasite density in the blood and the level of non-malaria infectious agents, respectively. Define the observed data vector $O_{ij} = (X_{ij},Z_{ij},G_{ij},Y_{ij})$. While $U_{ij}$ is generally not observed in malaria trials, $D_{ij}$ is. However, because the methods developed in this paper do not require or explicitly model parasite density we omit it from $O_{ij}$. We drop the subscripts and write $O = (X,Z,G,Y)$ to denote a random draw from a specified population. We will use boldface to indicate corresponding vector and matrix quantities that collect the data of subjects over each site or the entire trial. For example, $\Y_j$ is the $I_j\times 1$ vector that collects all the outcome data for subjects at site $j$ and $\Y = (\Y_1^T,\dots,\Y_J^T)^T$ is the $n\times 1$ vector that collects the outcome data for all subjects in the multi-site trial. For an example of a matrix quantity, $\X_j$ is an $I_j\times d$ matrix and $\X$ is a $n\times d$ matrix.

Multi-site RCTs, i.e., block randomized, are a common design for vaccine efficacy trials, such as the Phase III RTS,S trial \citep{rts2011}. Another design that has commonly been employed in studying the protective efficacy of interventions such as insecticide treated bed nets is the clustered RCT \citep{TER_KUILE_2003}. The notation is easily adapted for this design, letting $j$ indicate a cluster and $Z_{ij} = Z_j$ for all $i,j$. 

\subsection{Potential Outcomes and Malaria-Attributable Fever (or Death)}\label{subsec:PotOut}
\subsubsection{Pitfalls of Standard Case Definitions}\label{subsubsec:pitfalls}

The standard case definition of clinical malaria is the presence of a fever ($Y=1$) and a parasite density above some threshold $d$ ($D > d$). The WHO recommends setting $d$ high enough to achieve specificity $>80\%$ at all sites to avoid severe underestimation of vaccine efficacy \citep{MOORTHY20075115}. Sensitivity is also considered when determining $d$ to avoid under-powered studies. The prevailing method used to determine these thresholds is to model risk of fever as a continuous function of observed parasite density among community controls and clinically suspected cases \citep{smith1994attributable}. However, computing sensitivity and specificity of a case definition requires an estimate of the malaria attributable fraction of fevers (MAFF). \cite{lee2018estimating} show that obtaining unbiased estimates of MAFF in the presence of measurement error of $D$ and fever killing effects on parasite density is a difficult task, especially when malaria and non-malaria infections can work in conjunction to trigger a fever. Still, when unbiased estimates of MAFF can be obtained, this case definition will, by design, result in false positive and false negative cases, which may bias corresponding estimates of vaccine efficacy. Heterogeneous immunity and pyrogenic thresholds further complicate estimating treatment efficacy using case definitions that depend on a fixed threshold for $D$. In practice, standard thresholds without, it seems, careful estimation of specificity and sensitivity such as 2500 parasites per $\mu$l \citep{Olotu_2013} and 5000 parasites per $\mu$l are commonly used \citep{rts2011}. As far as we are aware, there are no sufficiently specific case-definitions for death attributable to malaria, the most severe effect of a malaria infection \citep{MOORTHY20075115}.

\subsubsection{Defining Cases Using Potential Outcomes}

Potential outcomes are a useful framework to precisely define causal effects of interest \citep{rub05}. The potential outcomes that we define now are closely related to those defined in \cite{lee2018estimating}. In what follows, we consider $Y$ to indicate the presence of fever but note that the framework we describe is also suitable when $Y$ indicates death. The aforementioned authors motivate their framework with a biological model of malaria and the notion of a pyrogenic threshold \citep{gravenor1998analysis}. A pyrogenic threshold can roughly be defined as a level of malaria parasite density above which a fever will be produced. Below the threshold, the infection will not be strong enough to trigger a fever. This threshold may vary from individual to individual based on heterogeneous immunity to symptomatic malaria. In general, we can think of $D$ and $U$ as having thresholds above which a malaria or non-malaria infection, respectively, is strong enough to trigger a fever in the absence of any other infection. We can also consider a curve of threshold pairs for $D$ and $U$ above which a combined infection can trigger a fever.

Because $Y$ indicates the presence of a fever resulting from any infection, we can treat $Y$ as a function of both $D$ and $U$. $Y$ can be thought of also as a function of $Z$ and $G$ through their effect on $D$ and $U$. Thus, we can write the potential outcome of $Y$ for treatment $z$ and sickle cell status $g$ as $Y(D(z,g),U(z,g))$ and the observed outcome as $Y = Y(D(Z,G),U(Z,G))$. $Y(D(z,g),U(z,g))$ factors additively into two natural terms,

\begin{equation}
  \underbrace{Y(D(z,g),U(z,g))}_{Y(z,g)} = \underbrace{\{Y(D(z,g),U(z,g)) - Y(0,U(z,g))\}}_{Y^m(z,g)}+ \underbrace{Y(0,U(z,g))}_{Y^{nm}(z,g)}\,.\label{eq:Ydecomp}
\end{equation}
$Y$ can be expressed as a potential outcome in terms of $d$ and $u$ or $z$ and $g$. When it is unambiguous, we may use the abbreviated notation below the ``curly'' braces in \eqref{eq:Ydecomp} to emphasize its dependence on $z$ and $g$. 

The statistical implication of a pyrogenic threshold is that $Y(d,u)$ is monotonic in $d$ for all $u$. That is, $Y(d,u) \le Y(d',u)$ for $0 \le d \le d'$. This ensures that the first term on the right hand side of \eqref{eq:Ydecomp} is non-negative. This term, $Y^m(z,g)$, can be interpreted as {\it malaria-attributable fever}, or a fever that would not have occurred had a malaria infection been absent. These include cases where $D$ was high enough to trigger a fever in the absence of any other infection and also cases where $D$ was high enough to trigger a fever in conjunction with a non-malaria infection. The second term on the right hand side, $Y^{nm}(z,g)$, is a fever that cannot be attributed to malaria. These are fevers that would have still occurred if malaria parasites were not present. However, when $D > 0$, $Y^m(Z,G)$ is generally not observable because it involves the counterfactual term $Y(0,U(Z,G))$. In the following section we will demonstrate that under a few realistic assumptions we can identify vaccine efficacy without directly observing the fevers (or deaths) attributable to malaria.

\subsubsection{A Final Couple Pieces of Notation}

With our potential outcome framework in place, we denote the complete data vector by \[C = (X,Z,G,D,U,\{Y^k(z,g):\, k =m,nm;\, z = 0,1;\, g = 0,1\})\,.\] In reality, we only get to observe one potential outcome $Y(Z,G)$ even though all four exist and, as we have argued, $Y^{nm}$ and $Y^m$ cannot be distinguished with certainty so we only observe $O \subset C$. Finally, we may sometimes treat $Y$, $D$ and $U$ as $p\times 1$ vectors that correspond to fever status, parasite density, and non-malaria infectious load at $p$ follow-up visits during a RCT.

\subsection{Vaccine Efficacy: Definitions, Assumptions and Identification}
\subsubsection{Defining Vaccine Efficacy}

We begin this section by defining a general population-level estimand for vaccine efficacy and then outline the assumptions required for identification when we cannot distinguish $Y^m$ from $Y^{nm}$. 
We suppose that for each $j$, $C_{ij}$ (and $O_{ij}$) are i.i.d. draws from an from an unknown, site-specific target population distribution $P_j\in\cP_j$ for all $i$. We suppose that the overall target population is an equally weighted mixture of the $P_j$ and denote it $P\in\cP$. We let $\mu_{zg}(P),\, \cP \to \bbR$ be a functional of $P$ that depends on treatment and sickle cell status and takes the form $\mu_{zg}(P) = \Exp{f\{Y(z,g)\}}{P}$ for $P$-measurable functions $f$ that are (1) increasing and (2) linear. 
We use a superscript $m$ to indicate that the functional is specific to malaria-attributable outcomes and $nm$ to indicate that it is specific to non-malaria outcomes, for instance, $\mu_{zg}^m(P) = \Exp{Y^m(z,g)}{P}$ when $f(y)=y$. Below we give a general definition of efficacy.

\begin{definition}[Vaccine/Protective Efficacy]\label{defn:efficacy}
  For a specified function $f$, 
  \begin{enumerate}[label=(\roman*)]
    \item Vaccine Efficacy is defined as $\tau(g) = 1 - \mu^m_{1g}(P)/\mu^m_{0g}(P)$ for $g = 0,1$ and
    \item Protective Efficacy of $G$ is defined as $\nu(z) = 1 - \mu^m_{z1}(P)/\mu^m_{z0}(P)$ for $z = 0,1$.
  \end{enumerate}
\end{definition}

We are primarily interested in two simple functions $f$: $f(y) = y$ and when $Y$ is a $p\times 1$ vector of fevers, $f(y) = \1^Ty$. For $f(y) = y$ , efficacy as defined in Definition \ref{defn:efficacy} is the proportion reduction in risk of malaria-attributable fever or death were an individual to receive vaccination versus placebo. For  $f(y) = \1^Ty$, efficacy is defined as the proportion reduction in incidence of malaria-attributable fever were an individual to receive vaccination versus placebo. Since you can only die once,  $f(y) = \1^Ty$ is not applicable when $Y$ indicates malaria-attributable death. By this same logic, $f(y) = y$ can be used for assessing the risk of malaria-attributable deaths over arbitrarily long follow-up. However, if we are interested in the risk of developing a malaria-attributable fever over a longer follow-up, during which a subject can have multiple fevers, the function we'd be interested in is $f(y) = \max y$ where $\max y$ is the maximum over the elements of a $p\times 1$ vector $y$. This function is not a linear function and thus vaccine efficacy against the risk of having at least one malaria-attributable fever over a long follow cannot be handled in this framework. The importance of the linearity of $f$ will become evident shortly.

\subsubsection{Identifying Vaccine Efficacy}\label{subsec:identify}

In \S \ref{sec:MFD_intro} we informally discussed the assumptions required to identify vaccine efficacy using MFD. Assumptions \ref{assump:add} - \ref{assump:valid} formalize the assumptions that are summarized in the second column of Table \ref{tab:parallelAssump} and provide the basis for the identification of $\tau$ proved in Proposition \ref{prop:identify}.

\begin{assumption}[Additivity of $\mu$] \label{assump:add}
  For all $z=0,1$ and $g=0,1$, $\mu_{zg}(P)$ can be decomposed linearly as
  $$ \mu_{zg}(P) = \mu^m_{zg}(P) + \mu^{nm}_{gz}(P)\,.$$
\end{assumption}
Assumption \ref{assump:add} is guaranteed by \eqref{eq:Ydecomp} and the linearity of $f$. We mentioned above that evaluating the vaccine efficacy on the risk scale for malaria-attributable fevers over a long follow-up is problematic when using MFD to estimate $\tau$. The risks of developing a malaria-attributable fever and a non-malaria fever do not satisfy Assumption \ref{assump:add} over long follow-ups because of the non-linearity of $f(y) = \max y$.

\begin{assumption}[No Interaction / Independent Protective Pathways]\label{assump:nointeract}
  $\tau(g)$ is constant over $g = 0,1$ and $\nu(z)$ is constant over $z=0,1$. That is, $\tau \coloneqq \tau(0)=\tau(1)$ and $\nu \coloneqq \nu(0)=\nu(1)$. 
\end{assumption}

The biological reasoning behind Assumption \ref{assump:nointeract} is as follows. If a vaccine and genetic trait protect against malaria through independent pathways at different times in the parasite's life cycle, it is plausible that the vaccine will prevent the same fraction of fevers among those who have the genetic trait ($G=1$) and those who do not ($G=0$). Similarly plausible is that possession of the genetic trait will prevent the same fraction of fevers among those to whom the vaccine was administered ($Z=1$) and to those it was not ($Z=0$). The goal of pre-erythrocytic vaccines like RTS,S is to provoke an immune response that prevents the parasites from entering the liver, stopping the parasites from ever re-entering the bloodstream and causing clinical symptoms \citep{regules2011rts}. In contrast, the sickle cell trait appears to protect against clinical symptoms by inhibiting the growth of parasites once they have re-entered the bloodstream from the liver and by making the host more tolerant to parasite infection \citep{ferreira2011sickle,taylor2012haemoglobinopathies,Williams_2011}. This suggests that Assumption \ref{assump:nointeract} is satisfied for pre-erythrocytic vaccines and the sickle cell trait.

\begin{assumption}[No Interference]\label{assump:nointerfere}
  A subject's potential outcomes are functions of its vaccine and sickle cell status alone. That is $Y_{ij}(\z,\g) = Y_{ij}(z_{ij},g_{ij})$.
\end{assumption}

Assumption \ref{assump:nointerfere} implies that the treatment and sickle cell status of an individual in the trial effects only their own outcome. This assumption can be rephrased in the context of infectious disease as stating that protection conferred by genetics or vaccine do not materially disrupt disease transmission. This assumption is plausible when considering individuals at two different sites $j\ne j'$ in a multi-site trial but is more complicated for individuals at the same site who may live in close proximity. Stochastic simulation models of malaria vaccination using pre-erythrocytic and blood stage-vaccines have, however, found that transmission effects of such vaccines delivered as they would be in pediatric malaria vaccination programs were minimal \citep{penny2008should,penny2015public}. Less is known about the suitability of Assumption \ref{assump:nointerfere} with respect to the sickle cell variant.

\begin{assumption}[Randomization / ``as-if'' Randomized]\label{assump:random}
    At each site $j=1,\dots,J$, vaccines are administered randomly and sickle cell status is distributed ``as-if'' random; that is, $$(Z,G) \indep (Y(z,g),Y^m(z,g),Y^{nm}(z,g),X)$$ for all $z,g \in \{0,1\}^2$. Additionally, we assume that $Z\indep G$.
\end{assumption}

The randomization of $Z$ is ensured by the RCT. We assume that within each center, the sickle cell trait is distributed ``as-if'' random.
Population stratification and linkage disequilibrium are common biological violations of the ``as-if'' random assumption of Mendelian genes \citep{Kang_2013}. Briefly, population stratification is when a subgroups that differ on prognostic factors for developing malaria and other childhood illnesses also systematically differ in the prevalence of HbAS. If either the probability of inheriting HbAS or the distribution of prognostic factors is relatively homogenous within a study site, then population stratification would likely not be a material threat to the ``as-if'' random assumption about $G$. Linkage disequilibrium is the dependence of gene frequencies at two or more loci \citep{MORTON20011105}. If HbAS is in linkage disequilibrium with a gene that affects the risk of childhood illness this may threaten the validity of Assumption \ref{assump:random}. Linkage disequilibrium with a gene that affects the risk of malaria is less problematic as this would not violate the exclusion restriction formalized in the next assumption. However, this would require a more delicate treatment of potential outcomes, e.g., we would need to consider the genes in linkage disequilibrium as having potential outcomes depending on $G$ (see \cite{VanderWeele_2013} for a discussion a related topic of ``versions" of treatment).

\begin{assumption}[Valid Mendelian Factor] \label{assump:valid}
  $G$ is a ``valid Mendelian factor'' in that it satisfies the following conditions:
  \begin{enumerate}[label=(\roman*)]
    \item $\nu \ne 0$; and
    \item $1 - \mu_{z1}^{nm}(P)/\mu_{z0}^{nm}(P) = 0$ for $z = 0,1$.
  \end{enumerate}
\end{assumption}

Part (i) of Assumption \ref{assump:valid} says that the Mendelian factor is {\it relevant} to malaria-attributable outcomes. A large body of literature on the protective properties of the sickle cell trait against malaria supports this assumption for HbAS. Several cohort studies have found evidence that HbAS has 30-50\% efficacy against uncomplicated clinical malaria and multiple other case-control and cohort studies of Africa have estimated even greater efficacies against sever malaria cases of 70-90\% (see \cite{Gong_2013} for a list of several studies reporting the protective efficacy of HbAS). Assumption \ref{assump:valid}(ii) can be expressed in terms of potential outcomes by the restriction that $U$ depends only on treatment status, $U(z,g)=U(z)$. The plausibility of Assumption \ref{assump:valid}(ii) is supported by evidence that the protection conferred by HbAS is ``remarkably specific'' to malaria, providing little protection to other childhood diseases \citep{williams2005sickle}.

The following proposition provides a simple, non-parametric identification strategy for treatment efficacy as defined in Definition \ref{defn:efficacy} under Assumptions \ref{assump:add} - \ref{assump:valid}. We also make the standard assumptions of {\it consistency}, that $Y=Y(z,g)$ when $Z=z$ and $G=g$, and {\it positivity}, that the probability of treatment $Z$ and the prevalence of $G$ in each site are both bounded away from $0$ and $1$.

\begin{proposition}[Nonparametric Identification] \label{prop:identify}
  Suppose that Assumptions \ref{assump:add} - \ref{assump:valid} are satisfied. Then the vaccine efficacy $\tau$ is identified from the observed data $\mathbf{O}$ as
\begin{equation}\label{eq:identify} \tau = 1 - \frac{\Exp{\Exp{f(Y)|\,X,Z=1,G=1}{P}}{X}-\Exp{\Exp{f(Y)|\,X,Z=1,G=0}{P}}{X}}{\Exp{\Exp{f(Y)|\,X,Z=0,G=1}{P}}{X}-\Exp{\Exp{f(Y)|\,X,Z=0,G=0}{P}}{X}}\,, \end{equation}
  where $\mathbb{E}_X$ is the expectation over the marginal distribution of $X$ implied by $P$.
\end{proposition}
\begin{proof}
  The proof of Proposition \ref{prop:identify} can be found in Appendix \ref{app:prop1} of the Supplemental Web Materials.
\end{proof} 

\subsubsection{Remarks}
Proposition \ref{prop:identify} still holds under a weaker version of Assumption \ref{assump:random} requiring only that $Z$ and $G$ are independent of potential outcomes (and of each other) {\it conditional} on baseline covariates $X$. Depending on how rich the set of baseline covariates $X$ is, this weaker assumption may be more tenable in the presence of population stratification and linkage disequilibrium.

Equation \eqref{eq:identify} suggests a ratio estimator for $\tau$ that is remarkably similar to the Wald estimator used in Mendelian randomization studies \citep{Wald_1940,burgess2017review}. Without covariates, the Wald estimator of the effect of a non-randomized exposure $E$ on outcome $Y$ using Mendelian gene $G$ can be written as 
\begin{equation}\label{eq:wald}
  \frac{\bar{Y}_{G=1}-\bar{Y}_{G=0}}{\bar{E}_{G=1}-\bar{E}_{G=0}}\,,
\end{equation}
where $\bar{V}_{G=g}$ is the sample average of $V$ for individuals with $G=g$. Both \eqref{eq:identify} and \eqref{eq:wald} involve ratios of averages differenced over $G$. The analogy between \eqref{eq:identify} and the Wald estimator in \eqref{eq:wald} is not by coincidence, but instead arises from a symmetry between the structures of Mendelian randomization with additive effects and Mendelian factorial design with multiplicative effects. More precisely, in Mendelian randomization, the potentially confounded association between the Mendelian gene and the outcome can be decomposed {\it multiplicatively} into the {\it additive} effect of the instrument on the exposure and the {\it additive} effect of the exposure on the outcome. In a Mendelian factorial design, the outcomes can be {\it additively} decomposed into disease-attributable outcomes and non-disease outcomes upon which the treatment and Mendelian factor have {\it multiplicative} effects. This simple structure of Mendelian factorial design can be seen clearly in the $2\times 2$ table in Figure \ref{fig:2x2} of expected outcomes $\mu_{zg}$ for different combinations of $z$ and $g$. Differencing over the columns and taking the ratio over the rows immediately yields $1-\tau$. The $\eta$ term in the top row can be thought of as the {\it spillover efficacy} of the vaccine against non-malaria outcomes. The term drops out when differencing over the columns. 

\begin{figure}[ht]
  \centering
  \renewcommand{\arraystretch}{2}
  \begin{tabular}{rcc}
    & $\mathbf{g=1}$ & $\mathbf{g=0}$ \\
    \cline{2-3}
    \multicolumn{1}{r|}{$\mathbf{z=1}$} &\multicolumn{1}{c|}{$\mu^{nm}(1-\eta)+\mu^m(1-\tau)(1-\nu)$} & \multicolumn{1}{c|}{$\mu^{nm}(1-\eta)+\mu^m(1-\tau)$}\\
    \cline{2-3}
    \multicolumn{1}{r|}{$\mathbf{z=0}$} & \multicolumn{1}{c|}{$\mu^{nm}+\mu^m(1-\nu)$} & \multicolumn{1}{c|}{$\mu^{nm}+\mu^m$}\\
    \cline{2-3}
  \end{tabular}
  \vspace{0.5cm}
  \caption[$2\times 2$ table for Mendelian factorial design.]{$2\times 2$ table for Mendelian factorial design. Each cell represents $\mu_{zg}(P)$ for all combinations of $(z,g) \in \{0,1\}^2$, which is identified by the observed data $\mathbf{O}$ (Proposition \ref{prop:identify}). Differencing over the columns then taking the ratio over the rows yields $1-\tau$. For notational clarity, we drop the subscript from $\mu_{00}$. $\eta$ is the ``spillover efficacy'' the vaccine may provide against non-malaria outcomes.}\label{fig:2x2}
\end{figure}

\subsection{Robust Covariance Adjusted Estimation and Inference}


We now propose a simple substitution estimator that allows for covariance adjustment and is robust to arbitrary misspecification of a model for $\Exp{f(Y)\,|\,X,Z,G}{P}$. The estimation procedure closely resembles that which is developed in \cite{Rosenblum_2010} with a couple minor modifications to deal with the factorial structure of our identification procedure and the possibility that HbAS prevalence varies across sites in a multi-site RCT. The simple procedure is detailed below in Algorithm \ref{alg:est} and requires estimating a simple generalized linear model (GLM)  with the \texttt{R} function \texttt{glm} in the \texttt{stats} package at most two times. In what follows, we suppose $f(y) = \1^Ty$ for expository purposes. That is, we will focus on vaccine efficacy defined as the proportion reduction in the expected number of malaria-attributable fevers.

\begin{algorithm}[Estimation of $\tau$] \label{alg:est}The following substitution estimator is based on \cite{Rosenblum_2010}. For clarity, we let $f(y) = \1^Ty$.

\noindent\rule[1ex]{\textwidth}{0.4pt}
	\begin{enumerate}[label=\arabic*.]
		\item Estimate $\Exp{f(Y)\,|\,X,\,Z,\,G}{P}$ with \texttt{glm} using a canonical link function (depends on $f$).
		\begin{itemize}
			\item Let the linear part include an intercept, main terms for $Z$ and $G$, and interaction $Z\times G$, for example,
			\begin{center}
			 \texttt{mu\_hat\_0 <- glm(f(Y) $\mathtt{\sim}$ X*G*Z, family = poisson())}
			 \end{center}
			\item denote the resulting estimator as $\hat\mu_0(Z,G,X)$. 
		\end{itemize}
		\item If there is more than one site and the prevalence of $G$ and sample sizes variy across sites, let $\texttt{log\_mu\_hat} = \log\hat{\mu}_0(Z_{ij},G_{ij},X_{ij})$, $\texttt{p\_g} = (1/I_j)\sum_{i=1}^{I_j} G_{ij}$, $\texttt{w} = n/I_j$, and \texttt{S} be a categorical variable for site; update $\hat\mu_0$ as follows
		\begin{align*}
			\texttt{mu\_hat\_1 <- }&\texttt{glm(f(Y) $\mathtt{\sim}$ offset(log\_mu\_hat) + S +}\\
			&\texttt{I(w*Z*G/p\_g) + I(w*Z*(1-G)/(1-p\_g)) + I(w*(1-Z)G/p\_g) + }\\
			&\texttt{I(w*(1-Z)*(1-G)/(1-p\_g)), family = poisson()) }
		\end{align*}
		 and denote the resulting estimator as $\hat\mu_1(Z,G,X)$.
		\item Let $\mu_{zg}(P_n) = \frac{1}{J}\sum_{j=1}^J\frac{1}{I_j}\sum_{i=1}^{I_j}\hat\mu_1(z,g,X_{ij})$. 
		\item Construct the plug-in MFD estimator  \[\hat\tau = 1 - \frac{\mu_{11}(P_n) - \mu_{10}(P_n)}{\mu_{01}(P_n) - \mu_{00}(P_n)}\,.\]
	\end{enumerate}
\noindent\rule[1ex]{\textwidth}{0.4pt}
\end{algorithm}

The estimator $\hat\tau$ returned by Algorithm \ref{alg:est} is a special case of a target maximum likelihood estimator (TMLE) \citep{van2011targeted}. A straightforward modification of Theorem 1 in \cite{Rosenblum_2010} yields that $\hat\tau$ is consistent and asymptotically normal under mild regularity conditions even when the working model for the conditional expectation is misspecified. Other initial working models $\hat\mu_0$ may be used as long as they satisfy certain restrictions on how data-adaptive they are \citep{van2011targeted}. The weight terms \texttt{w} are important because we assume that the target population $P$ is an equally weighted mixture of the site-specific populations $P_j$ but allow for study designs with different different sample sizes across sites. The weights ensure that $\hat\mu_1$ solves the efficient influence function estimating equation (see Appendix \ref{app:tmle} of the Supplemental Web Materials). Before we state the result, we introduce some some important notation and definitions.\\ 

\noindent{\it Working Models and their Limits:}

Let the maximum likelihood parameter estimates from steps 1 and 2 of Algorithm \ref{alg:est} be ${\boldsymbol{\beta}}^{(0)}_n$ and ${\boldsymbol{\beta}}^{(1)}_n$, respectively, and define ${\boldsymbol{\beta}}_n = [{\boldsymbol{\beta}}^{(0)}_n,{\boldsymbol{\beta}}^{(1)}_n]$. Now, recall that $P$ is the true, unknown data generating distribution. We can decompose its corresponding density $p$ as follows as follows
\begin{equation}
  p = p(X)p(Z)p_j(G)p(Y|\,X,\,Z,\,G)
\end{equation}
where $p_j(G)$ is the probability a child has the sickle cell trait in site $j$. $p(Z)$ is assumed to be known, e.g., $p(Z)=1/2$ in a balanced trial. $P_n$ is our estimate of $P$ where $p_n$, the density of $P_n$, can be decomposed similarly as
\begin{equation}
  p_n = p_n(X)p(Z)p_{j,n}(G)p_{\beta_n}(Y|\,X,\,Z,\,G)
\end{equation}
where $p_n(X)$ is the empirical distribution of $X$, $p_{j,n}(G)$ is the observed prevalence of the sickle cell trait among children in enrolled in the study at site $j$, and $p_{\boldsymbol{\beta}_n}(Y|\,X,\,Z,\,G)$ is the estimated parametric working model for the conditional distribution of $Y$ from steps 1 and 2 in Algorithm \ref{alg:est}. We assume that the number of centers are fixed, that they are representative of the population of interest, and that the sites carry equal weight in the population they represent but that the sample sizes $I_j$ might be different. Hence, the observations that make up the empirical distribution are weighted by $n/I_j$. We assume that the site-level sample sizes $I_j$ grow at the same rate as $n\to \infty$ and so it follows that $p_n(X) \stackrel{a.s.}{\to} p(X)$ as $n \to \infty$ by the Gilvenko-Cantelli Theorem and an application of the strong law of large numbers gives us that $p_{j,n}(G) \stackrel{a.s.}{\to} p_j(G)$  for all $j = 1,\dots,J$ as $n \to \infty$. Let $P_\infty = \lim_n P_n$ and write its density $p_\infty$ as
\begin{equation}
  p_\infty = p(X)p(Z)p_{j}(G)p_{\boldsymbol{\beta}}(Y|\,X,\,Z,\,G)
\end{equation}
where $\boldsymbol{\beta} = [\boldsymbol{\beta}^{(0)},\boldsymbol{\beta}^{(1)}]$ are the maximizers of the expected log-likelihoods of the GLMs in steps 1 and 2 where the expectation is taken over $P$, if such maximizers exist (see \cite{Rosenblum_2010} for further discussion of the existence of $\boldsymbol{\beta}$). When $\boldsymbol{\beta}$ exists, the conditions given in Proposition \ref{prop:tmle},  are sufficient for $\boldsymbol{\beta}_n$ to converge to $\boldsymbol{\beta}$ in probability \citep{rosenblum2009using}. Note that unless the working parametric model for the conditional mean is correctly specified, $P_\infty$ will not equal $P$ in general. 

For notational convenience, we let $\mathbb{P}$ and $\mathbb{P}_n$ be the expectation operators over $P$ and $P_n$, respectively. For example, we can write $\mu_{zg}(P_n)$ from step 3 of Algorithm \ref{alg:est} as $\mathbb{P}_n\hat\mu_1(z,g,X)$.\\

\noindent{\it Efficient Influence Functions:}

For an arbitrary distribution $Q \in \cP$, the {\it efficient influence function} (EIF) for $\mu_{zg}(Q)$ can be written as
\begin{align}
  \varphi_{zg}(Q)(O) & = \frac{\mathbbm{1}(Z=z)\mathbbm{1}(G=g)(f(Y) - \Exp{f(Y)|\,X,\,Z=z,\,G=g}{Q})}{q_{j}(G=g)q(Z=z)}\notag\\
  &\qquad+ \Exp{f(Y)|\,X,\,Z=z,\,G=g}{Q} - \mu_{zg}(Q)\,,
\end{align}
for $z,g \in \{0,1\}^2$. We will define $\mu_{z}(Q) \coloneqq \mu_{z1}(Q)-\mu_{z0}(Q)$ and $\varphi_{z}(Q) \coloneqq \varphi_{z1}(Q)-\varphi_{z0}(Q)$ for $z=0,1$. Standard calculations verify that $\varphi_{z}(Q)$ is the efficient influence function for $\mu_{z}(Q)$ by demonstrating that $\varphi_{z}(Q)$ can be expressed as a pathwise derivative of $\mu_{z}(Q_\epsilon)$ where $Q_\epsilon$ is a parametric submodel of $Q$ such that $Q=Q_{\epsilon=0}$ \citep{kennedy2016semiparametric}. $\varphi_{z}(Q)$ is said to be a pathwise derivative of $\mu_{z}(Q_\epsilon)$ if $\Exp{\varphi_{z}(Q)S_\epsilon}{Q} = \partial\mu_{z}(Q_\epsilon)/\partial\epsilon|_{\epsilon=0}$ where $S_\epsilon$ is the score function of the parametric submodel. Finally, let $\varphi^*_{zg}(Q) = \varphi_{zg}(Q) - \Exp{\varphi_{zg}(Q)\,|\,G}{P}$ and define $\varphi_{z}^*(Q) \coloneqq \varphi^*_{z1}(Q)-\varphi^*_{z0}(Q)$. We can now state the consistency and asymptotic normality result for $\hat\tau$.

\begin{proposition}[Consistency and Asymptotic Normality]\label{prop:tmle}
	In addition to Assumptions \ref{assump:add} - \ref{assump:valid}, suppose that the number of sites $J$ is fixed, the $I_j$ grow at the same rate as $n$, and the maximizers ${\boldsymbol{\beta}}=[\boldsymbol{\beta}^{(0)},\boldsymbol{\beta}^{(1)}]$ exist. Also, suppose that $\pnorm{\boldsymbol{\beta}}_{\infty} < M$ for pre-specified $M < \infty$. Finally, let $(X,Y)$ be bounded and the terms in the linear parts of the GLMs in steps 1 and 2 of Algoritm \ref{alg:est} be bounded functions on compact subsets of $\{0,1\}^2\times\mathbb{R}^d$. Then, $\hat\tau$ is consistent and $\sqrt{n}(\hat\tau-\tau)$ convergence in distribution to a Gaussian with mean $0$ and variance
	\begin{equation}
  	   \sigma^2 = \Exp{\frac{\mu_1(P)}{\mu_0(P)^2}(\varphi^*_{0}(P_\infty)(O) - \mathbb{P}\varphi^*_{0}(P_\infty)(O)\} - \frac{1}{\mu_0(P)}\{\varphi^*_{1}(P_\infty)(O) - \mathbb{P}\varphi^*_{1}(P_\infty)(O)\} }{P}^2\,.\notag
	\end{equation} 
	When the prevalence of $G$ does not vary between sites and the study is balanced, i.e., $I_j = n/J$ for all $j$, then step 2 of Algorithm \ref{alg:est} can be skipped. Step 2 may also be skipped if it is a single-site trial. Furthermore, if the working model for $\Exp{f(Y)\,|\, X,Z,G}{P}$ is correctly specified then $\sigma^2$ achieves the semiparametric efficiency bound.
\end{proposition}

\begin{proof}
  A sketch of the proof of Proposition \ref{prop:tmle} can be found in Appendix \ref{app:tmle} of the Supplemental Web Materials.
\end{proof}

\begin{remark}
	We remarked in \S \ref{subsec:identify} that Proposition \ref{prop:identify} still holds under a weaker version of Assumption \ref{assump:random} where the independence is conditional on baseline covariates $X$. Proposition \ref{prop:tmle} also holds under these weaker assumptions as long as either the working model for the conditional expectation or the conditional assignment models for $G$ and $Z$ are correctly specified.
\end{remark}

\noindent{\it Variance Estimator:}

Notice that $\varphi^*_z$ is defined as a projection of $\varphi_z$ into a lower dimensional subspace and will thus have smaller variance. With that in mind, we can use $\varphi_z(P_n)$ to construct a conservative plug-in estimator of $\sigma^2$. In Appendix \ref{app:tmle} of the Supplemental Web Materials we show that $P_n$ solves the EIF estimating equations, that is, $\mathbb{P}_n\varphi_z(P_n)(O) = 0$ for $z=0,1$. It follows then that the plug-in estimator of the scale variance $\sigma^2$ can be written simply as

\begin{equation}
  n\cdot\widehat\var(\hat\tau) = \frac{1}{n}\sum_{j=1}^J\sum_{i=1}^{I_j} \left\{ \varphi_{0}(P_n)(O_{ij})\frac{\mu_{1}(P_n)}{\mu_{1}(P_n)^2} - \varphi_{1}(P_n)(O_{ij})\frac{1}{\mu_{0}(P_n)} \right\}^2
\end{equation} 

With this variance estimator we can now use Proposition \ref{prop:tmle} to conduct inference on and construct confidence intervals for $\hat\tau$.\\

\noindent{\it Naive Estimator:}

For simplicity, let's assume that we have balanced sites, i.e., $I_j=n/J$ for all $j$. Then, had we assumed that all fevers with any parasitemia were malaria-attributable fevers, we might considered the following {\it naive estimator}
\begin{equation}
  \hat\tau_0 = 1 - \frac{\mathbb{P}_n \hat\mu_0(1,G,X)}{\mathbb{P}_n \hat\mu_0(0,G,X)}\,.
\end{equation}
The naive estimator corresponds with standard estimates of VE with respect to a commonly used secondary case definition of the presence of a fever ($Y=1$) and any positive parasite density ($D>0$) \citep{Olotu_2013}.  We will see in the following section that the estimator is systematically biased but can be combined with our MFD estimator $\hat\tau$ to construct an estimator that outperforms both estimators on their own.

\subsection{Simulation Study: Comparison to Naive Identification Strategy}\label{subsec:simulation_irr}
In this section we investigate the performance of our proposed MFD estimator $\hat\tau$ and compare it to that of the naive estimator $\hat\tau_0$ that assume all fevers with any parasitemia are malaria fevers. As in the previous section, we consider $f(y) = \1^Ty$ and thus $\tau$ is the proportion reduction in expected number of malaria-attributable fevers. We consider a single-site RCT $J=1$ with equal sized vaccine and placebo arms and the prevalence of HbAS set to 20\% based on existing estimates of the prevalence in sub-Saharan Africa \citep{TER_KUILE_2003, elguero2015malaria}. 

We simulate the number of malaria-attributable fevers and non-malaria fevers from negative binomial distributions in a single year of follow up for each individual. Evidence suggests that the negative binomial distribution fits the empirical distribution of the number of clinical malaria events an individual experiences well \citep{Olotu_2013}. We can also see from \eqref{eq:Ydecomp} that $Y^m$ and $Y^{nm}$ are negatively dependent  -- a fever cannot be both attributable to malaria and have been present in the absence of malaria. To model this dependence structure, we use a Gaussian copula with negative dependence parameter $\rho=-0.1$ \citep{genest2007primer}. The negative dependence is modest when fevers are rare but will be more pronounced in areas where fever risk is higher, e.g., villages with poor sanitation. We suppose there is a single observed covariate $X$ that enters the conditional mean function for both $\1^TY^m$ and $\1^TY^{nm}$ and that there is unobserved heterogeneity in the conditional means between individuals. The generating distributions were calibrated so that the average number of any-cause fevers is 1.5 per child-year. We also calibrated the generating distributions to achieve different levels of case specificity, which we define formally below.

Two different sample sizes were chosen to assess how performance improved asymptotically: 1000 subjects per trial arm, i.e., $n = 2000$, and 2500 subjects per trial arm, i.e., $n=5000$. For $n=2000$, we consider $5\times 2\times 2$ different combinations of vaccine efficacy $\tau \in \{0.3,0.4,0.5,0.6,0.7\}$, protective efficacy of the Mendelian factor $\nu \in \{0.3,0.5\}$, and specificities $s\in\{0.5,0.8\}$. Formally, we define specificity $s$ here as the expected number of malaria-attributable fevers divided by the expected number of fever with any parasitemia and of any cause under placebo,

\begin{equation}
	s = \frac{\Exp{\1^TY^m(0,G)}{P}}{\Exp{\1^TY(0,G)}{P}} = \frac{p(G=1)\mu_{01}^m + p(G=0)\mu_{00}^m}{p(G=1)\mu_{01} + p(G=0)\mu_{00}}\,.
\end{equation}

 The specificity choices are motivated by \cite{mabunda2009country}, which estimated the specificity of standard case definitions using $>0$ and $>2500$ parasites per $\mu$l cutoffs for children under five years of age to be roughly $50\%$ and $80\%$, respectively. The strength of the Mendelian factor was calibrated to the aforementioned estimates of the protective efficacy of HbAS. We consider the same combinations for $n=5000$ except only for the weaker protective efficacy setting. The spillover efficacy $\eta$ was assumed to be zero. Finally, we allow for non-constant vaccine and Mendelian protective efficacy -- both $\tau_i$ and $\nu_i$ are log-normally distributed with mean $\tau$ and $\nu$, respectively, and standard deviation approximately $0.05$ giving coefficients of variation of about 8-17\%. Rather than conferring complete protection to a fraction of the vaccinated subjects, all vaccinated subjects receive partial protection. This is consistent with evidence that RTS,S is a ``leaky" vaccine, providing at least partial protection to all recipients of the vaccine \citep{moorthy2009immunological}. Each setting was simulated $N_{sim}=5000$ times. More details of the simulation settings can be found in Appendix \ref{app:simdetail} of the Supplemental Web Materials. \texttt{R} code to reproduce the simulation results can be found in the Supplementary Materials.

We use Poisson regression for the initial working model estimate in Algorithm \ref{alg:est} and because $J=1$, we skip step 2. The results of the simulation study are summarized in the following two tables. In Table \ref{tab:bias_rmse_irr} we compare the absolute proportional bias and root mean squared error (RMSE) of the $\hat\tau$ (MFD) and $\hat\tau_0$ (Naive). For $n=2000$, the MFD estimator has very good bias properties, with proportional bias less than 5\% for most settings. The RMSE increases for Mendelian genes that are less protective and as the specificity decreases. The only area of poor bias performance is when  the Mendelian gene is weakly protective ($\nu = 0.3$) and the specificity is low (0.5). However, although the estimator in this setting has high variance the bias and power at $\tau=0.7$ are reasonably adequate. Like the presence of small sample bias and high variance for weak instruments in instrumental variable analysis \citep{imb05}, the poor performance appears to be a small sample property as the performance for the weak Mendelian gene, low specificity setting improves notably for $n=5000$. 

The bias in the naive estimator only varies over the different specificity settings and does not improve as the sample size grows. In fact, because there is no spillover efficacy, the proportional absolute bias is equal to $1- s$. You can see from the table that the RMSE is driven almost entirely by the bias component for the naive estimator and it actually increases in absolute terms as the vaccine efficacy increases. These properties lead to very poor coverage of confidence intervals derived from the naive estimator. That one should expect efficacy estimates to be biased by as much as 20\% when specificity is at the level recommended by the WHO should be cause for concern.

\begin{table}\caption{Proportional absolute bias and root mean squared error (RMSE) of MFD and naive estimators using $N_{sim}=5000$ simulations.\label{tab:bias_rmse_irr} \vspace{0.25cm}}
  \ra{1.2}
  \centering
  \begin{tabular}{rrcccccccccccc}
    \toprule
    && \multicolumn{5}{c}{Specificity $=0.8$} && \multicolumn{5}{c}{Specificity $=0.5$}\\
    \cmidrule{3-7}\cmidrule{9-13}
     & & \multicolumn{2}{c}{Prop. $|$Bias$|$} & \phantom{a} & \multicolumn{2}{c}{RMSE} & \phantom{ab} & \multicolumn{2}{c}{Prop. $|$Bias$|$} & \phantom{a} & \multicolumn{2}{c}{RMSE} \\
    \cmidrule{3-4}\cmidrule{6-7}  \cmidrule{9-10}\cmidrule{12-13}
     $\nu$ & $\tau$ & MFD & Naive && MFD & Naive && MFD& Naive && MFD & Naive \\
    \midrule
    $\mathit{n=1000 (\times 2)}$\\ 
0.50 & 0.30 & 0.03 & 0.20 &  & 0.15 & 0.07 &  & 0.12 & 0.50 &  & 0.29 & 0.15 \\ 
   & 0.40 & 0.02 & 0.20 &  & 0.14 & 0.08 &  & 0.07 & 0.50 &  & 0.27 & 0.20 \\ 
   & 0.50 & 0.02 & 0.20 &  & 0.12 & 0.10 &  & 0.05 & 0.50 &  & 0.25 & 0.25 \\ 
   & 0.60 & 0.02 & 0.20 &  & 0.11 & 0.12 &  & 0.04 & 0.50 &  & 0.24 & 0.30 \\ 
   & 0.70 & 0.01 & 0.20 &  & 0.10 & 0.14 &  & 0.02 & 0.50 &  & 0.22 & 0.35 \vspace{0.2cm}\\ 
  0.30 & 0.30 & 0.12 & 0.20 &  & 0.36 & 0.07 &  & 0.59 & 0.50 &  & 5.47 & 0.15 \\ 
   & 0.40 & 0.11 & 0.20 &  & 0.30 & 0.09 &  & 0.32 & 0.50 &  & 4.54 & 0.20 \\ 
   & 0.50 & 0.08 & 0.20 &  & 0.29 & 0.10 &  & 0.21 & 0.50 &  & 5.80 & 0.25 \\ 
   & 0.60 & 0.06 & 0.20 &  & 0.24 & 0.12 &  & 0.18 & 0.50 &  & 4.61 & 0.30 \\ 
   & 0.70 & 0.03 & 0.20 &  & 0.21 & 0.14 &  & 0.11 & 0.50 &  & 1.69 & 0.35 \\ 
   $\mathit{n=2500 (\times 2)}$\\
  0.30 & 0.30 & 0.06 & 0.20 &  & 0.18 & 0.06 &  & 0.18 & 0.50 &  & 0.37 & 0.15 \\ 
   & 0.40 & 0.04 & 0.20 &  & 0.17 & 0.08 &  & 0.13 & 0.50 &  & 0.34 & 0.20 \\ 
   & 0.50 & 0.02 & 0.20 &  & 0.14 & 0.10 &  & 0.08 & 0.50 &  & 0.30 & 0.25 \\ 
   & 0.60 & 0.02 & 0.20 &  & 0.13 & 0.12 &  & 0.06 & 0.50 &  & 0.29 & 0.30 \\ 
   & 0.70 & 0.01 & 0.20 &  & 0.12 & 0.14 &  & 0.04 & 0.50 &  & 0.27 & 0.35 \\ 
    \bottomrule
  \end{tabular}
\end{table} 

Table \ref{tab:cov_pow_irr} below compares the coverage of two-sided 95\% confidence intervals and the power against the two-sided alternative at 5\% significance for the MFD and naive procedures. The MFD confidence interval has correct or conservative coverage and decent power in even some of the more unfavorable settings. Even in the small sample, weakly protective Mendelian gene, and low specificity setting the power is not negligible at higher vaccine efficacies. Because the naive estimator is systematically biased, the coverage properties are extremely poor in all settings. The variance of the naive estimator tends to be far smaller than the MFD estimator. This fact coupled with the systematic bias leads to high power and low coverage across all settings.

\begin{table}\caption{Coverage of two-sided $95\%$ confidence interval and power against two-sided alternative at $5\%$ significance level of MFD and naive estimators using $N_{sim}=5000$ simulations.\label{tab:cov_pow_irr} \vspace{0.25cm}}
  \ra{1.2}
  \centering
  \begin{tabular}{rrcccccccccccc}
    \toprule
    && \multicolumn{5}{c}{Specificity $=0.8$} && \multicolumn{5}{c}{Specificity $=0.5$}\\
    \cmidrule{3-7}\cmidrule{9-13}
     & & \multicolumn{2}{c}{Coverage} & \phantom{a} & \multicolumn{2}{c}{Power} & \phantom{ab} & \multicolumn{2}{c}{Coverage} & \phantom{a} & \multicolumn{2}{c}{Power} \\
    \cmidrule{3-4}\cmidrule{6-7}  \cmidrule{9-10}\cmidrule{12-13}
     $\nu$ & $\tau$ & MFD & Naive && MFD & Naive && MFD& Naive && MFD & Naive \\
    \midrule
    $\mathit{n=1000 (\times 2)}$\\
0.50 & 0.30 & 0.95 & 0.50 &  & 0.56 & 1.00 &  & 0.96 & 0.00 &  & 0.29 & 0.99 \\ 
   & 0.40 & 0.95 & 0.17 &  & 0.79 & 1.00 &  & 0.96 & 0.00 &  & 0.44 & 1.00 \\ 
   & 0.50 & 0.96 & 0.01 &  & 0.94 & 1.00 &  & 0.96 & 0.00 &  & 0.59 & 1.00 \\ 
   & 0.60 & 0.96 & 0.00 &  & 0.99 & 1.00 &  & 0.96 & 0.00 &  & 0.73 & 1.00 \\ 
   & 0.70 & 0.96 & 0.00 &  & 1.00 & 1.00 &  & 0.97 & 0.00 &  & 0.84 & 1.00  \vspace{0.2cm} \\ 
  0.30 & 0.30 & 0.95 & 0.50 &  & 0.30 & 1.00 &  & 0.95 & 0.00 &  & 0.19 & 0.99 \\ 
   & 0.40 & 0.95 & 0.16 &  & 0.44 & 1.00 &  & 0.96 & 0.00 &  & 0.25 & 1.00 \\ 
   & 0.50 & 0.96 & 0.01 &  & 0.57 & 1.00 &  & 0.96 & 0.00 &  & 0.33 & 1.00 \\ 
   & 0.60 & 0.96 & 0.00 &  & 0.74 & 1.00 &  & 0.97 & 0.00 &  & 0.42 & 1.00 \\ 
   & 0.70 & 0.96 & 0.00 &  & 0.86 & 1.00 &  & 0.98 & 0.00 &  & 0.51 & 1.00 \\ 
   $\mathit{n=2500 (\times 2)}$\\
  0.30 & 0.30 & 0.95 & 0.11 &  & 0.45 & 1.00 &  & 0.95 & 0.00 &  & 0.26 & 1.00 \\ 
   & 0.40 & 0.96 & 0.00 &  & 0.68 & 1.00 &  & 0.96 & 0.00 &  & 0.38 & 1.00 \\ 
   & 0.50 & 0.95 & 0.00 &  & 0.87 & 1.00 &  & 0.96 & 0.00 &  & 0.52 & 1.00 \\ 
   & 0.60 & 0.96 & 0.00 &  & 0.96 & 1.00 &  & 0.97 & 0.00 &  & 0.66 & 1.00 \\ 
   & 0.70 & 0.96 & 0.00 &  & 0.99 & 1.00 &  & 0.97 & 0.00 &  & 0.77 & 1.00 \\ 
    \bottomrule
  \end{tabular}
\end{table} 

Even in settings where the mean proportional bias is high, the MFD estimator appears to have the desirable property of being median unbiased. Figure \ref{fig:med_unbiased} demonstrates this for $\tau = 0.3,0.5,$ and $0.7$ in all six settings detailed in Tables  \ref{tab:bias_rmse_irr} and \ref{tab:cov_pow_irr}. The dashed lines indicates the true value of $\tau$, the boxes represents the IQRs, and the whiskers indicate the 5\% and 95\% quantiles of the simulated estimates. The diamonds indicate the means and the vertical solid lines the medians. The worsening performance of the MFD as the Mendelian gene weakens is clear (settings 2 vs. 1, 4 vs. 3, and 6 vs. 5) as is the improvement in the weak Mendelian gene settings as the sample size grows (settings 5 vs. 3 and 6 vs. 4). The naive estimates have much lower variance but are systematically mean and median biased downward. The MFD estimates are also mean biased downward. Although the estimator is asymptotically normal,  the ratio of means that are jointly asymptotically normal may have a peculiar non-normal form in finite samples which may explain this particular pattern of bias \citep{marsaglia2006ratios}.

\begin{figure}[ht]
\centering
\includegraphics[scale=0.5]{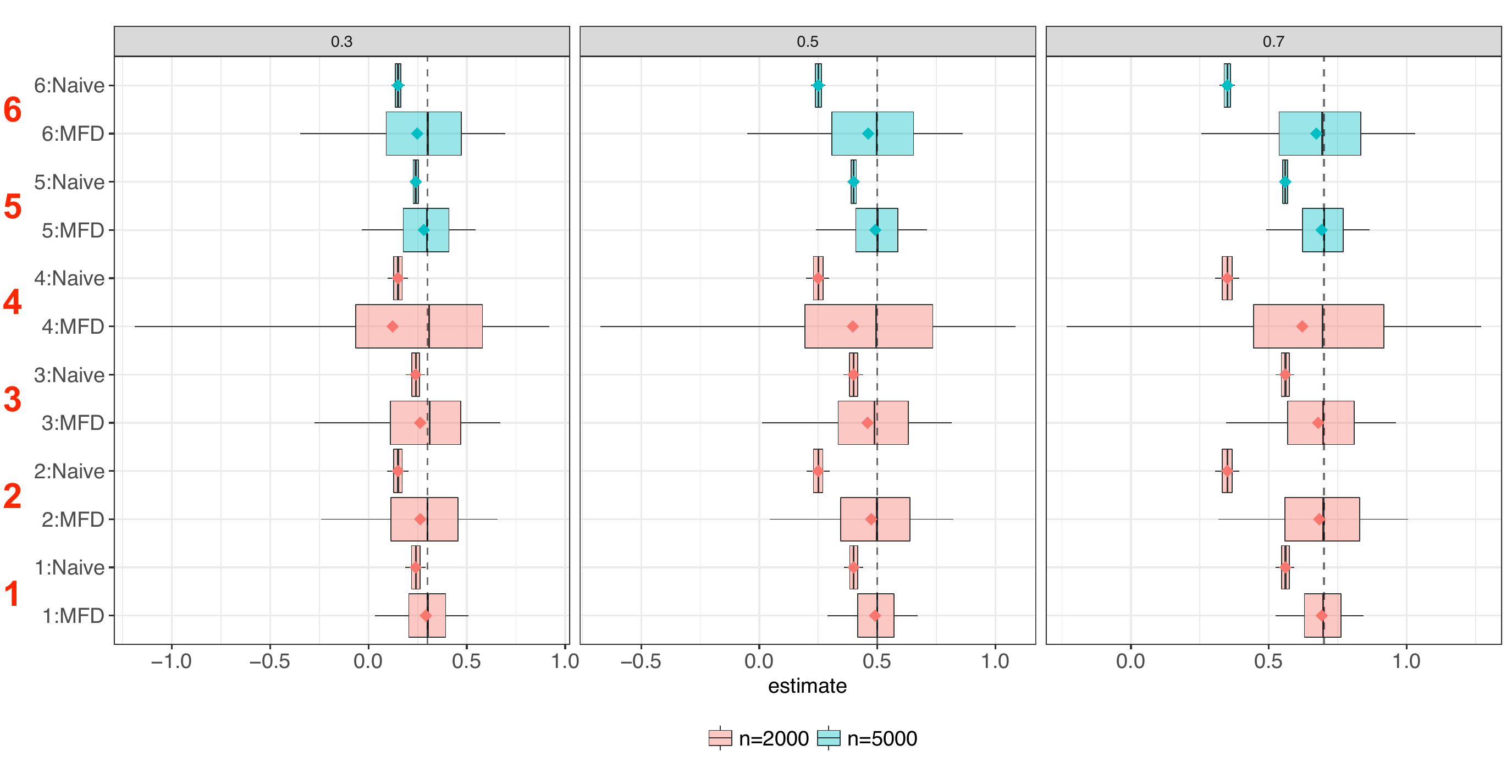} 
\caption[Distributions of simulated MFD estimator $\hat\tau$ and naive estimator $\hat\tau_0$ over several settings.]{Distributions of simulated MFD estimator $\hat\tau$ and naive estimator $\hat\tau_0$ over several settings and $\tau=0.3,0.5,0.7$. Setting \textbf{1}: strong factor ($\nu=0.5$), high specificity ($s=0.8$); setting \textbf{2}: strong factor, low specificity ($s=0.5$); settings \textbf{3} and \textbf{5}: weak factor ($\nu=0.3$), high specificity; settings \textbf{4} and \textbf{6}: weak factor, low specificity. Dashed lines indicate true efficacy, boxes indicate IQRs, whiskers are 5\% and 95\% quantiles, diamonds are mean estimate, and vertical solid lines are median estimates.}\label{fig:med_unbiased}
\end{figure}

\subsection{Improved Estimators: Leveraging the Identifying Assumptions}

\subsubsection{A Simple Correction to the Naive Estimator}

The observation that the proportional absolute bias of $\hat\tau_0$ is equal to $1-s$ in Table \ref{tab:bias_rmse_irr} for all settings comes from the fact that, when the sample sizes are balanced across sites, an immediate application of Theorem 1 of \cite{Rosenblum_2010} yields that $\hat\tau_0$ is a consistent estimator of $s\tau + (1-s)\eta$. When samples are small or the Mendelian gene is only weakly protective, we observed that the MFD estimator will be less powerful and may suffer from small sample bias. In such settings, the asymptotic limit of $\hat\tau_0$ suggests a simple correction to the naive estimator: dividing by $s$. Call this the $s$-{\it corrected estimator}, which is consistent for $\tau$ when $\eta=0$ and $s$ is known or consistently estimated itself.  If instead we only have a $1-\beta$ confidence interval for $s$, $\cC_\beta$, then we can still construct a valid confidence interval for $\tau$ using the $s$-corrected estimator. Let $\text{CI}_{s,\alpha+\beta}$ be a $1-\alpha-\beta$ confidence interval for $\tau$ constructed using the $s$-corrected estimator and define the $s$-corrected $1-\alpha$ confidence interval $\text{CI}_{corr,\alpha} = \bigcup_{s\in \cC_\beta} \text{CI}_{s,\alpha+\beta}$. \cite{berger1994} show that one can construct valid p-values by maximizing a non-pivotal p-value over the confidence set of a nuisance parameter. This result is easily inverted, providing a procedure to construct valid confidence intervals from which it follows that $\text{CI}_{corr,\alpha}$ will have the correct (conservative) coverage for $\tau$. However, as we mentioned earlier in \S \ref{subsubsec:pitfalls}, estimating and conducting inference about $s$ is challenging. However, we can still make use of the naive estimator even when we do not know $s$. 

\subsubsection{The Best of Both Worlds? Combining MFD and Naive Estimators}

The naive estimator and the MFD estimator have complementary strengths and different weaknesses -- $\hat\tau_0$ tends to be more efficient but is asymptotically biased and $\hat\tau$ is consistent but has higher variance and requires larger sample sizes. The naive estimator also has the nice feature of being a consistent lower bound of $\tau$ as long as $\eta \le \tau$ -- it is very plausible that a well designed vaccine will have a higher vaccine efficacy than spillover efficacy. Additionally, we have the logical constraint that $\tau\le 1$ since a treatment efficacy greater than one would imply that it would be possible to have a negative number of malaria-attributable fevers, a scenario eliminated by the assumption that $Y(d,u)$ is monotonically increasing in $d$ for all levels $u$. The following proposition constructs an estimator that uses these upper and lower bounds to improve the performance of $\hat\tau$ in difficult settings, e.g., small samples, weak Mendelian factor, and low specificity.

\begin{proposition}[Bounded Estimator]\label{prop:bnd_estimator}
Suppose that $\eta \le \tau$. Let $\hat\tau_{0}$ be the naive estimator and let
\begin{align*}
	L_{\alpha} & = \hat\tau - \Phi(1-\alpha) \widehat\var(\hat\tau)^{1/2}\\
	L_{0,\alpha} &= \hat\tau_0 - \Phi(1-\alpha)  \widehat\var(\hat\tau_0)^{1/2}\,.\\
\end{align*}
Define the upper confidence bounds $U_{\alpha}$ and $U_{0,\alpha}$ similarly. Then (i)
\begin{equation}
	\hat\tau_{bnd} = \min\left[1,\max\left\{\hat\tau,L_{0,\tilde\alpha}\right\}\right]
\end{equation}
is consistent for $\tau$ when $\tilde\alpha$ is bounded away from $1$; and (ii), for $0 \le \alpha_0 \le \alpha/2$
\begin{equation}
	\text{CI}_{bnd,\alpha} = \left[\max\{L_{\alpha/2-\alpha_0},L_{0,\alpha_0}\},\min\left\{1,U_{\alpha/2}\right\}\right]
\end{equation}
is an asymptotically valid $1-\alpha$ confidence interval for $\tau$.
\end{proposition}

\begin{proof}
	The proof of Proposition \ref{prop:bnd_estimator} can be found in Appendix \ref{app:prop2} of the Supplemental Web Materials.
\end{proof}

Because the naive estimator has much smaller variance than the MFD estimator, one will likely choose $\alpha_0$ and $\tilde\alpha$ to be much smaller than $\alpha$.  In general, it makes sense to choose $\tilde\alpha=\alpha_0$ to ensure that $\hat\tau_{bnd} \in \text{CI}_{bnd,\alpha}$. 

The lower bound used to construct $\hat\tau_{bnd}$ acts like a high probability, stochastic lower bound to $\tau$ while $1$ is an exact upper bound. When the Mendelian gene is weak, the denominator in $\hat\tau$ can sometimes be very close to zero, leading to unreliable estimates of vaccine efficacy. The bounded estimator is designed to mitigate the effect of these cases while keeping the median unbiasedness of $\hat\tau$ intact. Figure \ref{fig:bnd_dens}, illustrates how this plays out in a the setting where the $\hat\tau$ performs poorly with $n=2000$, a weak Mendelian factor, low specificity, and $\tau = 0.5$. The figure shows the estimated densities and means of the naive estimator (blue, dash), MFD estimator (gray, dot), and bounded estimator (green, dot-dash) over 5000 simulations. The true vaccine efficacy is indicated by the black solid line. As expected, $\hat\tau$ is much more variable and less biased than $\hat\tau_0$, but still materially biased in this setting. The bounded estimator $\hat\tau_{bnd}$ is nearly mean unbiased. You can see how it achieves this by ``clumping" unreliable MFD estimates near the stochastic lower bound and the exact upper bound.

\begin{figure}[ht]
\centering
\includegraphics[scale=0.475]{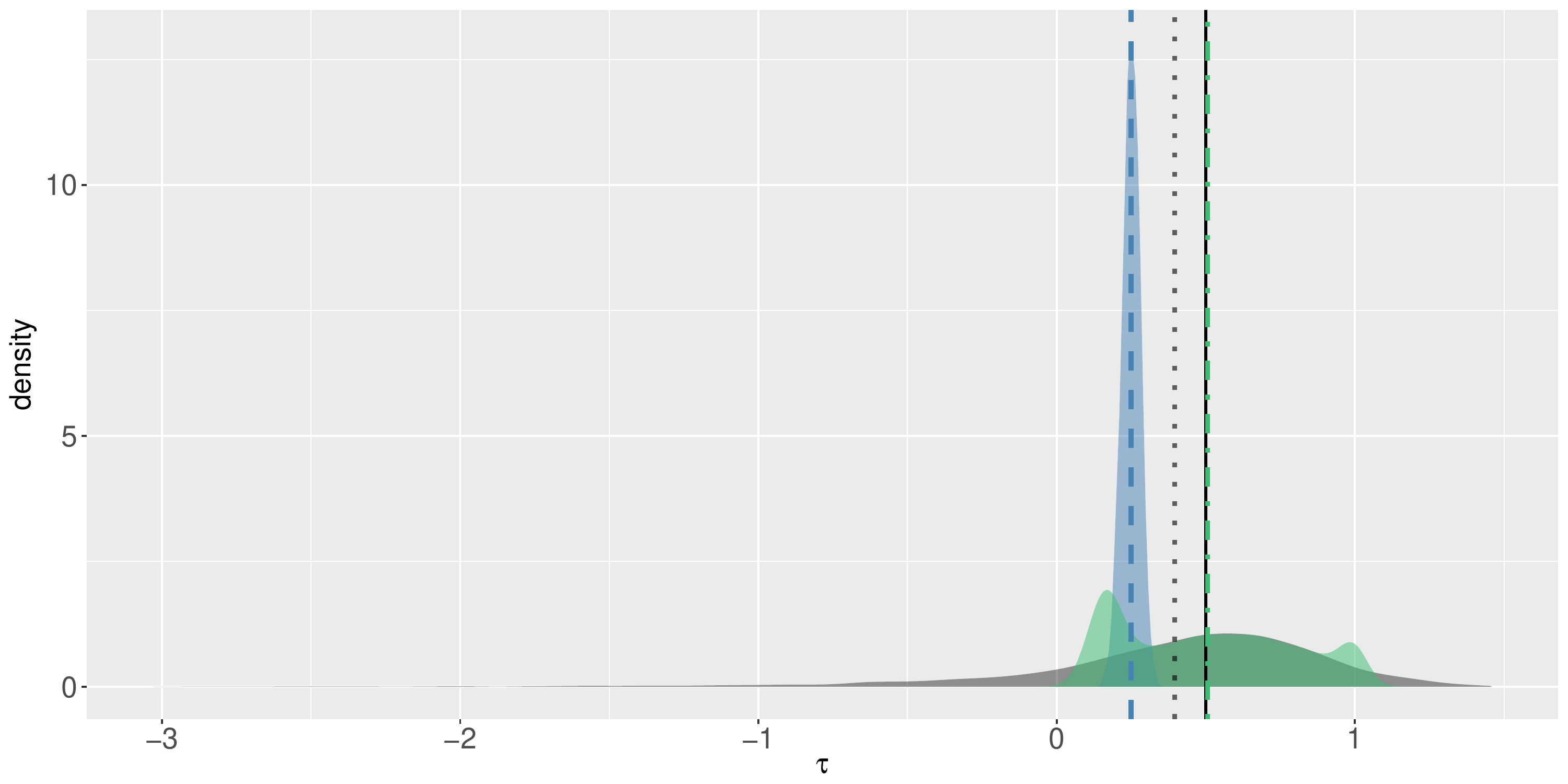}
\caption[Comparing densities and means of $\hat\tau_0$, $\hat\tau$, and $\hat\tau_{bnd}$ across 5000 simulations.]{Densities and means of naive estimator (blue, dash) MFD estimator (gray, dot), and bounded estimator (green, dot-dash);  True vaccine efficacy (black, solid). For the setting with $n=1000(\times 2)$, $\tau = 0.5$, $\nu=0.3$, $\text{spec.}=0.5$, and $\tilde\alpha = 0.001$}\label{fig:bnd_dens}
\end{figure}

We also assess how well $\hat\tau_{bnd}$ performs in the same setting as in Figure \ref{fig:bnd_dens} but with an even smaller sample size ($n=1000$). Using $\alpha=0.05$, $\alpha_0 = 0.001$, and $\tilde\alpha=0.001$, we find that $\hat\tau_{bnd}$ has substantially improved bias and RMSE while $\text{CI}_{bnd,\alpha}$ has very favorable power while maintaining the correct (conservative) coverage. In Figure \ref{fig:bias_rmse_bnd}, we compare the absolute proportional bias (left panel) and RMSE (right panel) of the three estimators.  Absolute proportional bias and RMSE values larger than 1 are not shown. The bounded estimator uniformly outperforms both the MFD and naive estimators in terms of bias and has performance comparable to that of the MFD estimator estimated on a sample five times as large ($n=5000$) for $\tau = 0.5,0.6,\text{ and }0.7$. The RMSE of $\hat\tau_{bnd}$ shows large improvements over the $\hat\tau$ and is comparable to the RMSE of $\hat\tau$ estimated in the same setting on a sample size of $n=5000$ for all values of $\tau$ considered. 

\begin{figure}[ht!]
\centering
\begin{tabular}{cc}
\includegraphics[scale=0.5]{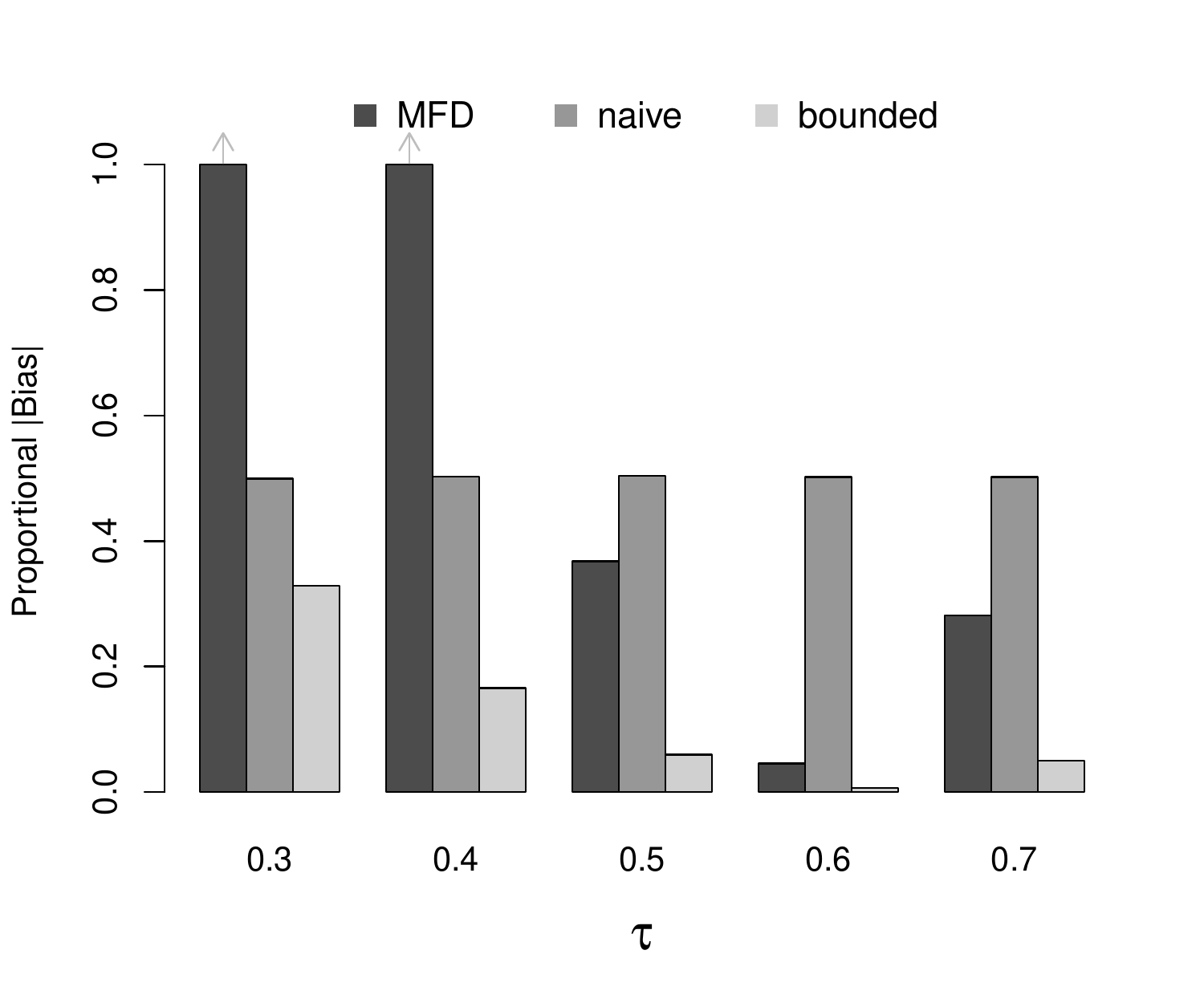} & \includegraphics[scale=0.5]{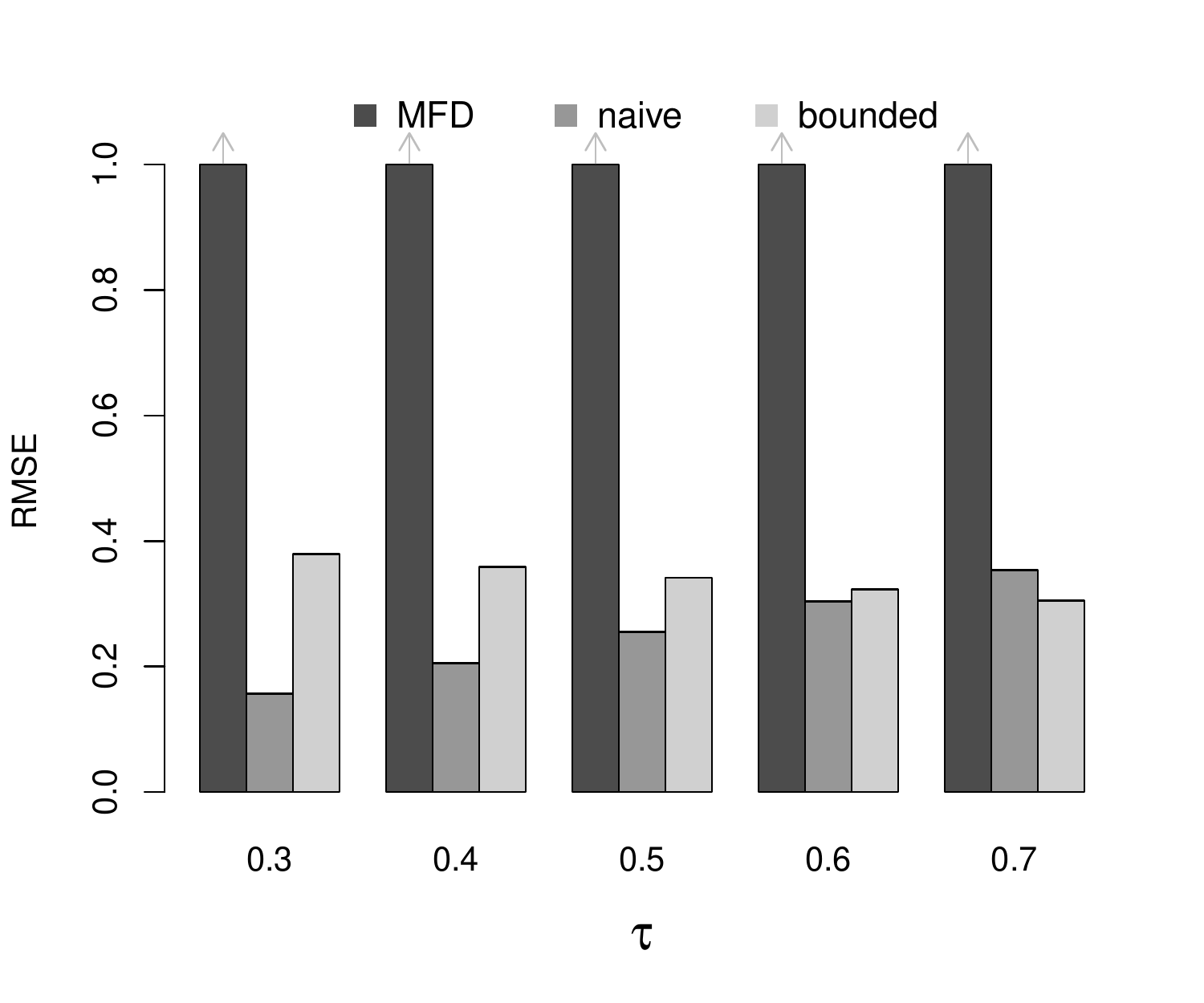} \\
\end{tabular}\caption[Absolute proportional bias and RMSE for MFD, naive, and bounded estimators.]{Absolute proportional bias (left panel) and RMSE (right panel) for MFD, naive, and bounded estimators with sample size $n=1000$ and $N_{sim} = 5000$ simulations. For bias and RMSE values above $1$, only a maximum of $1$ is shown. The actual absolute proportional bias values for the MFD estimator are $2.21$ and $1.56$ for $\tau = 0.3$ and $0.4$, respectively. The actual RMSE values for the MFD estimator from left to right are $18.58,67.90,32.29,8.42,\text{ and }6.75$.}\label{fig:bias_rmse_bnd}
\end{figure}

The power and coverage properties of $\text{CI}_{bnd,\alpha}$ give the clearest picture of how the complementary strengths of the MFD and naive procedures are retained by the bounded procedure. In the left panel of Figure \ref{fig:pow_cov_bnd}, we see that $\text{CI}_{bnd,0.05}$ has only marginally more conservative coverage than the MFD confidence interval $[L_{0.025},U_{0.025}]$ (right panel) while retaining the favorable power properties of the naive confidence interval $[L_{0,0.025},U_{0,0.025}]$ (left panel). The dashed line in the right panel indicates $95\%$ coverage.

\begin{figure}[ht!]
\centering
\begin{tabular}{cc}
\includegraphics[scale=0.5]{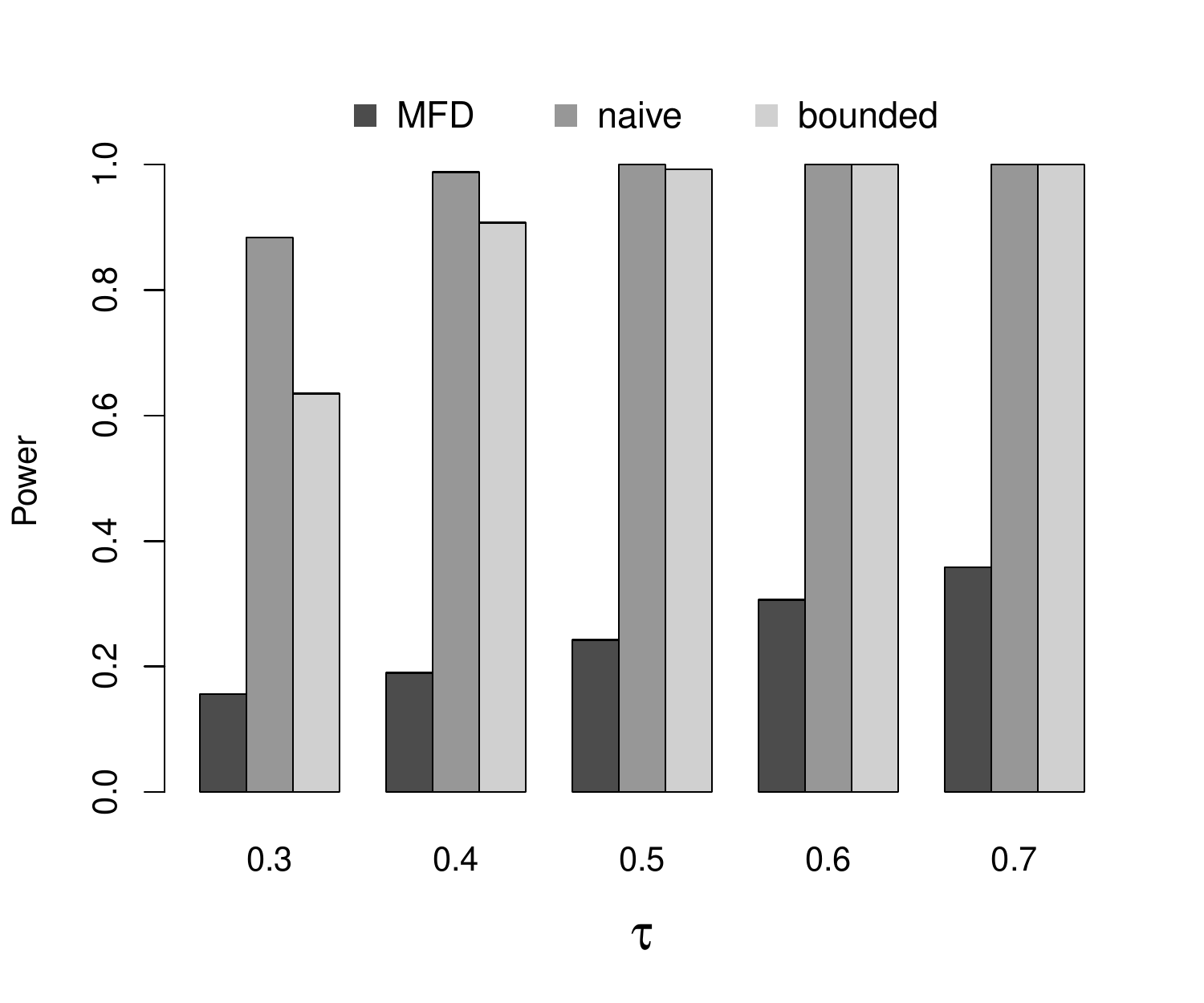} & \includegraphics[scale=0.5]{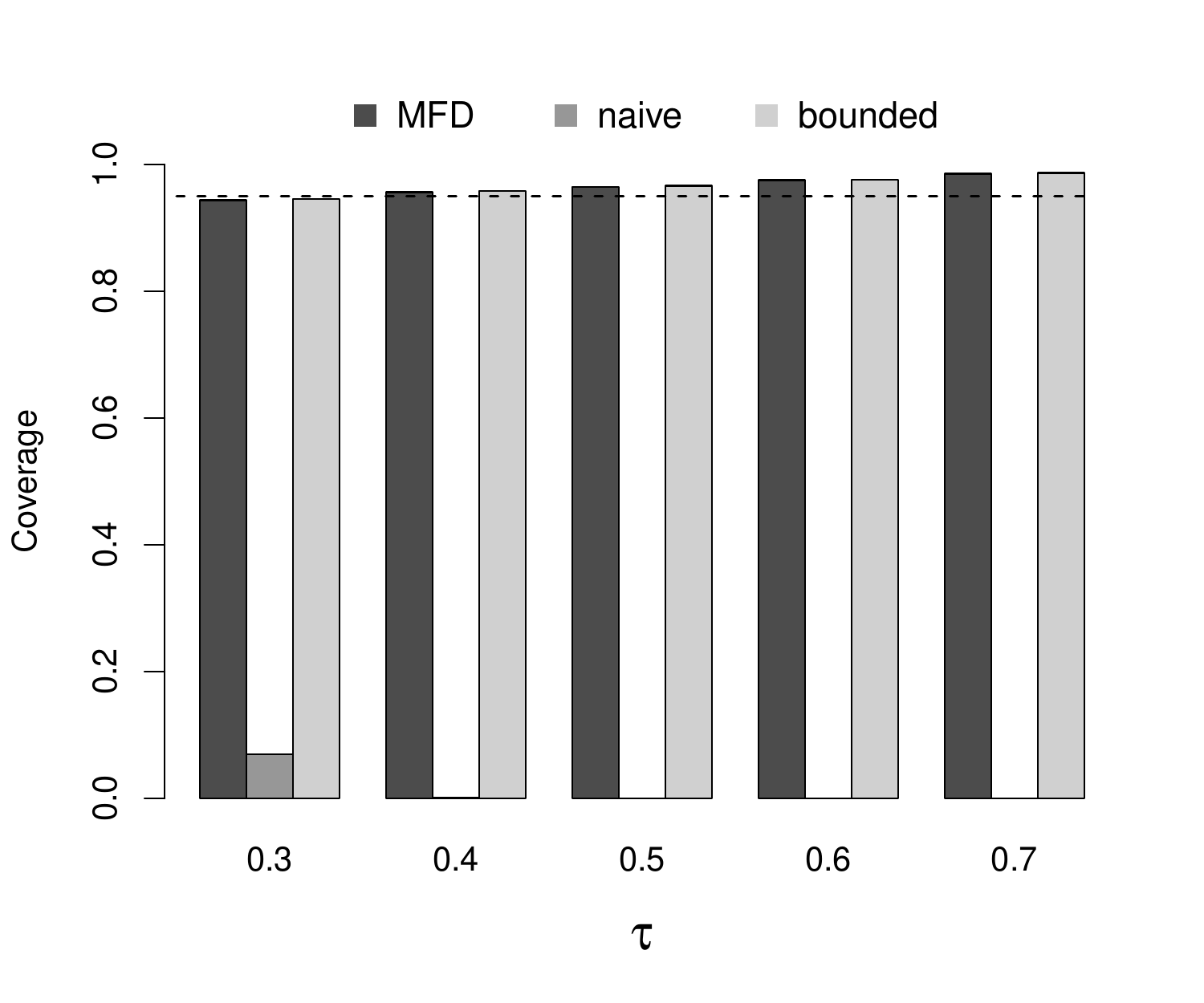} \\
\end{tabular}\caption[Power against two-sided alternative and coverage for MFD, naive, and bounded estimators]{Power against two-sided alternative at $5\%$ significance (left panel) and coverage of $95\%$ confidence interval (right panel) for MFD, naive, and bounded estimators with sample size $n=1000$ and $N_{sim} = 5000$ simulations. The dashed black line in the right panel indicates $95\%$ coverage.}\label{fig:pow_cov_bnd}
\end{figure}

Importantly, the superior performance of the bounded estimator in this particular setting does not appear to come at the expense of performance in the settings where $\hat\tau$ generally does well, improving on the $\hat\tau$ in almost all settings investigated in the simulation study.  

\section{Discussion} 

The strategy that we've developed for identification and estimation of malaria vaccine efficacy does not rely on the explicit definition of an inexact, but observable case definition. In short, we propose separating the approach to identifying VE into two distinct steps: first, define a gold-standard case that may be unobservable and then, identify VE using a strategy that doesn't require overt observation of these gold-standard cases. Our gold-standard case definition is one that is 100\% specific and sensitive. It is precisely defined using a potential outcome framework. As we've noted, these cases are not distinguishable from observed data alone. Regardless, we demonstrate that with Mendelian factorial design we are able to leverage genetic variation in an analogous fashion to Mendelian randomization studies to identify VE with respect to this exact, but unobservable, case definition. 

In observational studies, {\it evidence factors} are defined as several approximately independent tests of the same hypothesis, each of which depend on different assumptions about bias from non-random treatment assignment \citep{rosenbaum2011some}. In the presence of bias, these tests can be thought of as independent pieces of evidence in that the violation of assumptions underlying one test does not imply the other tests are similarly biased. Like an evidence factor, the MFD estimate can provide an additional piece of evidence in vaccine efficacy studies that relies on a different set of assumptions than the methods currently in use. For instance, the naive estimator assumes all fevers with any parasitemia are malaria-attributable, corresponding to a commonly used secondary case definition. Although not independent like true evidence factors, we demonstrated that the naive estimator and the MFD estimator can be combined to construct a {\it bounded} estimator that outperforms both estimators when used on their own. In particular, the bounded estimator provides significant improvements when the Mendelian gene is only weakly protective, a challenging setting similar to the weak instrument setting in IV studies. 

There is evidence that HbAS is to moderately protective against uncomplicated malaria and highly protective against severe malaria illness suggesting that MFD may be particularly useful for estimating VE against severe malaria, whose symptoms are not specific and overlap significantly with symptoms of other severe childhood comorbidities \citep{bejon2007defining}.

The performance of the combined estimator across a range of simulation settings with sample sizes that are similar to those found in phase II and III clinical trials is promising.  The results suggests that, if feasible, it would be prudent to begin collecting subject data on inherited hemoglobinopathies and other genetic traits that provide protection against clinical malaria, such as the sickle cell trait. Estimating the efficacy of prevention strategies that target transmission directly, such as a malaria transmission-blocking vaccine \citep{wu2015development} and insecticide-treated bed nets \citep{TER_KUILE_2003}, are feasible under the MFD framework developed in this paper. The assumption of no interference will likely fail but weakening Assumption \ref{assump:nointerfere} to allow for partial interference \citep{hudgens2008toward}, e.g., interference within but not between villages or sites, is plausible. Cluster randomized trial designs would allow for identification of natural definitions of treatment efficacy that are functions of the fraction of subjects who receive treatment \citep{athey2018exact}. The study of erythrocytic vaccines in the MFD framework may be more challenging, as the assumption that protective hemoglobinopathies and these blood-stage vaccines don't interact is tenuous at best. With more than thirty vaccines currently under development, both pre-erythrocytic and those targeting different stages of the disease cycle, the methods described in this paper have the potential to improve the reliability of vaccine efficacy estimates for a number of forthcoming trials \citep{mahmoudi2017efficacy}.

Many interesting research directions related to Mendelian factorial design remain. We mention a few in closing. Developing MFD methods using only aggregate site-level data on HbAS prevalence is one such direction. This might allow for MFD-based meta-analyses of past trial results in which hemoglobinopathy data were not collected if, (1) accurate prevalence data could be collected ex post and (2) the prevalence varied sufficiently between studies. There is evidence that there are a number of other genetic traits that confer protection against malaria \citep{ndila2018human}. When there are several potential valid Mendelian factors, we may again find inspiration from MR and other IV methods. It would be beneficial to develop falsification tests of the validity of potential Mendelian factors akin to tests for over-identifying restrictions such as the Sargan-Hansen test in IV regression \citep{hansen1982large}. As it has been demonstrated in MR, using many ostensible Mendelian factors where some are invalid may have the potential to improve the robustness of MFD analysis \citep{kang2016instrumental}. Finally, the application of MFD to the study of vaccines against other childhood diseases that have unspecific symptoms and no gold-standard case definition, such as pediatric tuberculosis \citep{beneri2016understanding}, could prove fruitful if plausible mendelian factors can be identified.
\newpage
\begin{center}
{\Large\bf Supplemental Web Materials for \textit{Estimating Malaria Vaccine Efficacy in the Absence of a Gold Standard Case Definition: Mendelian Factorial Design}}
\end{center}
%
%
%
%
%
\vspace{0.5cm}
\renewcommand{\thesection}{\Alph{section}}
\setcounter{section}{0} 

\section{Time-to-First Malaria Fever: Mendelian Factorial Design Under The Proportional Hazards Assumption}\label{sec:ttffev}
The World Health Organization (WHO) recommends that the primary endpoint in pivotal Phase III trials assessing the efficacy of malaria vaccines be the time-to-first malaria fever and that the efficacy be measured as one minus the hazard ratio returned by a Cox proportional hazard regression \citep{MOORTHY20075115}. In this section we use our potential outcome framework to precisely define vaccine efficacy in time-to-first malaria fever studies and show how it can be identified using MFD and estimated using data on time-to-first fever of any cause.
\subsection{Parameter Identification and Estimation under Modeling Assumptions}\label{subsec:identificationPH}
In this section we can think about our clinical outcomes as stochastic processes indexed by $t$, $Y \coloneqq \{Y_t\}_{t\ge 0}$. This allows us to consider quantities like the time to first fever, $T = \min\{t:\,Y_t = 1\}$.

If we follow the WHO recommendation above in the context of the potential outcome framework developed in \S \ref{subsec:PotOut}, a natural quantity to study in a time-to-event analysis of vaccine efficacy is the hazard rate of malaria-attributable fevers, i.e., the instantaneous risk of developing a malaria-attributable fever at time $t$ conditional on being free of malaria-attributable fevers up to time $t$. The risk set used in this hazard function includes individuals who have experienced non-malaria fevers prior to time $t$. This is analogous to a {\it subdistribution hazard} in the competing risks literature where non-malaria fevers can be viewed as a competing risk \citep{fine1999proportional}. Unfortunately, when $Y^m$ is not observable, vaccine efficacy based on a hazard ratio using the definition of malaria-attributable fever in \eqref{eq:Ydecomp} is not identifiable using MFD. This is due to the fact that on the hazard ratio scale, Assumption \ref{assump:add} (additivity) does not hold when $Y^m$ and $Y^{nm}$ are dependent -- recall, if $Y_t^{nm}$ equals $1$ then $Y_t^{m}$ must equal $0$ but is otherwise free to take values in $\{0,1\}$. Only under independence (or independence conditional on $X$) can we additively decompose the hazard of developing a fever of any cause into the hazard of developing a malaria-attributable fever and the hazard of developing a non-malaria fever. That said, under a few additional assumptions we can still identify vaccine efficacy in terms of the hazard rate for malaria-attributable fevers in the {\it absence of non-malaria infections}. This hazard is analogous to a cause-specific hazard in the competing risks literature when competing risks can be considered conditionally independent \citep{hsu2017statistical}.

\subsubsection{Malaria-attributable Fevers in the Absence of Competing Infections}

If we assume that there is no combined effect of malaria and non-malaria infections on fevers (Assumption 2(iii) of \cite{lee2018estimating}), then we can provide an alternative decomposition of $Y$ to the decomposition presented in \eqref{eq:Ydecomp}:

\begin{equation}
  Y(D(z,g),U(z)) = Y(D(z,g),0) \vee Y(0,U(z))\,,
\end{equation}

where $Y(D(z,g),0)$ are malaria-attributable fevers in the absence of non-malaria infections, $Y(0,U(z))$ are non-malaria fevers in the absence of malaria infections, and $\vee$ is a pairwise maximum. For stochastic processes, $A\vee B = \{\max(A_t,B_t)\}_{t\ge 0}$. We will refer to $Y(D(z,g),0)$  and $Y(0,U(z))$ as {\it isolated} malaria fevers and {\it isolated} non-malaria fevers. If $X$ sufficiently captures shared risk factors for malaria and non-malaria infections, then it is plausible that $D(z,g) \indep U(z) \;|\; X$ and thus $Y(D(z,g),0) \indep Y(0,U(z)) \;|\; X$ for all $z,g$. 

One might wonder if isolated malaria fevers are the endpoint of greatest interest. Perhaps malaria-attributable fevers as defined in \eqref{eq:Ydecomp} are more representative of the real-world burden of malaria infections. Regardless, one could argue that the standard case definition of clinical malaria, e.g., $Y=1$ and $D>2500$ parasites per $\mu l$, is an approximation of an isolated malaria fever. This argument has two parts: (1) the fixed threshold suggests that the standard case definition does not consider the possibility of combined effects of malaria and non-malaria fevers; and (2) the standard case definition does not consider whether a non-malaria fever would have occurred had there been no malaria infection. The approximation is rough, however, because while the definition of isolated malaria fever allows for individual-specific pyrogenic thresholds, the standard case definition does not.

\subsubsection{Identifying Vaccine Efficacy}

We can now define the potential time to first isolated malaria fever as $T^m(z,g) = \min\{t:\,Y_t(D(z,g),0)=1\}$ and the potential time to first isolated non-malaria fever as $T^{nm}(z,g) = \min\{t:\,Y_t(0,U(z))=1\}$. We can write the conditional hazard functions for $T^m(z,g)$ and $T^{nm}(z,g)$ as 

\begin{equation}\label{eq:obsHaz}
  \lambda_{zg}^{k}(t\,|\,X) = \lim_{\Delta\to 0_+}\Pr{T^k(z,g)\in[t,t+\Delta\,)|\,T^k(z,g)\ge t,\,X}/\Delta
\end{equation}
for $k=m,nm$ and all $z,g$. We drop the superscript $k$ to indicate the conditional hazard function for the a fever of any cause and define $$T(z,g) = \min\{t:\,Y_t(z,g)=1\} =\min\{T^m,T^{nm}\}\,.$$ 
\begin{assumption}[Proportional Hazards]\label{assump:prop}
  $\lambda_{zg}^k(t\,|\,X)$, $k=m,nm$ follow a Cox proportional hazard model with baseline hazard functions $\lambda^k(t)$, $k=m,nm$. That is,
  \begin{enumerate}[label=(\roman*)]
    \item $\lambda_{zg}^m(t\,|\,X) = \lambda^m(t)\exp\{\log \kappa + \log (1-\tau) z + \log (1-\nu) g + \boldsymbol{\beta}_m^TX\}\,,$ and 
    \item $\lambda_{zg}^{nm}(t\,|\,X) = \lambda^{nm}(t)\exp\{\log \phi + \log (1-\eta) z + \boldsymbol{\beta}_{nm}^TX\}$
  \end{enumerate}
  where $\nu > 0$.
\end{assumption}
In equation (i) above, vaccine efficacy $\tau$ is equal to one minus the hazard ratio of isolated malaria fever under vaccination versus placebo. In (ii), $\eta$ can again be thought of as a ``spillover efficacy'' term. Note that Assumptions \ref{assump:nointeract} and \ref{assump:valid} are satisfied under these modeling assumptions.

\begin{assumption}[Conditionally Independent First Fever Processes]\label{assump:indepFev}
The time to first isolated malaria fever and time to first isolated non-malaria fever are conditionally independent. That is,
  \[T^m(z,g) \indep T^{nm}(z) \;|\; X\quad \text{for}\; (z,g) \in \{0,1\}^2\]
\end{assumption}

Because $T^m(z,g)$ depends only on $Y(D(z,g),0)$ and $T^{nm}(z,g)$ depends only on $Y(0,U(z))$, Assumption \ref{assump:indepFev} is satisfied when $Y(D(z,g),0) \indep Y(0,U(z)) \;|\; X$ for all $z,g$. We previously argued that this condition is plausible when $X$ sufficiently describes shared risk factors for malaria and non-malaria infections.

\begin{assumption}[Shared Conditional Baseline Hazard Function]\label{assump:shared}
  Time-to-first isolated malaria fever and time to first isolated non-malaria fever have a shared baseline hazard function conditional on $X$. That is, $\lambda(t) \coloneqq \lambda^m(t) = \lambda^{nm}(t)$  for all $t\ge 0$ and $\boldsymbol{\beta} \coloneqq \boldsymbol{\beta}_m = \boldsymbol{\beta}_{nm}$.
\end{assumption}

Under these additional assumptions, we can use an MFD strategy to identify $\tau$ in time-to-first fever studies.

\begin{proposition}[Identification of $\tau$]\label{prop:PHidentify}
  Suppose that Assumptions \ref{assump:prop} - \ref{assump:shared} hold. Then the hazard function for $T(z,g)$ can be written as
  \begin{equation}
    \lambda_{zg}(t\,|\,X) = \lambda(t)\exp\{\alpha + \omega z + \gamma g + \lambda z\times g + \boldsymbol{\beta}^TX\}\label{eq:cox_model}\,.
  \end{equation}
  Furthermore, under Assumptions \ref{assump:nointerfere} and \ref{assump:random}, $\tau$ is identified from the observed data $\mathbf{O}$ as
  \begin{equation}\label{eq:cox_tau}
    \tau = 1 - \frac{\exp\{\omega + \gamma + \lambda\}-\exp\{\omega\}}{\exp\{\gamma\}-1}\,.
  \end{equation}
\end{proposition}

\begin{proof}
  The proof of Proposition \ref{prop:PHidentify} can be found in Appendix \ref{app:prop3}
\end{proof}

\begin{remark}
  The conditions of Proposition \ref{prop:PHidentify} lead to an additivity property of the conditional hazard functions that is analogous to Assumption \ref{assump:add}. This additivity can be seen in the first equality of \eqref{eq:factorialPH}.
\end{remark}

\begin{remark}
  A weaker version of the randomization / ``as-if'' randomized condition imposed by Assumption \ref{assump:random} would suffice to identify $\tau$. Namely, that $Z,G$ are independent of all potential outcomes conditional on $X$.
\end{remark}

\subsubsection{Estimation and Inference via the Delta Method}\label{subsubsec:estdeltaPH}

The parameters in \eqref{eq:cox_model} can be estimated consistently via maximum partial likelihood estimation \citep{cox1975partial} with the \texttt{R} function \texttt{coxph} implemented in the \texttt{survival} package. These estimates $(\hat\omega,\hat\gamma,\hat\lambda)$ are asymptotically 
normal with the limiting distribution
\begin{equation}\label{eq:cox_normal}
	\sqrt{n}\left(\begin{bmatrix}
	\hat\omega\\ 
	\hat\gamma\\
	\hat\lambda 
	\end{bmatrix} - \begin{bmatrix}
	\omega\\ 
	\gamma\\
	\lambda 
	\end{bmatrix}\right) \;\stackrel{\mathcal{D}}{\longrightarrow}\; \cN\left(\begin{bmatrix}
	0\\
	0\\
	0
	\end{bmatrix},\,\text{I}([\omega,\gamma,\lambda^T])^{-1}\right)\,,
\end{equation}
where $\text{I}([\omega,\gamma,\lambda^T])^{-1}$ is the inverse of the partial likelihood-based information matrix. Applying the continuous mapping theorem to \eqref{eq:cox_tau} yields the following consistent MFD estimator of $\tau$,
\begin{equation}\label{eq:cox_est}
	 \hat\tau = 1 - \frac{\exp\{\hat\omega + \hat\gamma + \hat\lambda\}-\exp\{\hat\omega\}}{\exp\{\hat\gamma\}-1}\,.
\end{equation}
Applying the delta method to \eqref{eq:cox_normal} and \eqref{eq:cox_est} and noting that $\hat{\text{I}}([\hat\omega,\hat\gamma,\hat\lambda^T])^{-1}/n$, the average sample information, is consistent for $\text{I}([\omega,\gamma,\lambda^T])^{-1}$ gives us an approximate distribution for $\hat\tau$ for large enough samples,

\begin{equation}\label{eq:cox_inf}
	\hat\tau \;\dot\sim\; \cN\left(\tau, \boldsymbol{\partial\hat\tau}^T\, \hat{\text{I}}([\hat\omega,\hat\gamma,\hat\lambda^T])^{-1}\, \boldsymbol{\partial\hat\tau} \right)\,,
\end{equation}
where $\boldsymbol{\partial\hat\tau} = [\partial\hat\tau/\partial\hat\omega,\partial\hat\tau/\partial\hat\gamma,\partial\hat\tau/\partial\hat\lambda]^T$. We can use this approximate distribution to conduct inference on $\tau$. 

An analogous bounded estimator and confidence interval to those described in Proposition \ref{prop:bnd_estimator} can be constructed using the naive estimator $1-\exp\{\hat\omega\}$. It can easily be shown that this estimator is consistent for a convex combination of $\tau$ and $\eta$ and is thus is a consistent estimator for a lower bound of $\tau$ as long as $\eta \le \tau$. 

\subsection{Causal Identification: Selection Bias and Frailty}

Even if these strong modeling assumption assumptions hold, $\tau$ defined in \S \ref{subsec:identificationPH} does not have a straightforward causal interpretation. Several authors have discussed in depth the subtleties and common misconceptions of interpreting hazard ratios as a causal quantity \citep{hernan2010hazards,aalen2015does,martinussen2018subtleties}. We observe one of these subtleties when examining the definition of $\tau$ given in Assumption \ref{assump:prop}(i),
\begin{align}
  \tau & = 1 - \frac{\lim_{\Delta\to 0_+}\Pr{T^m(1,g)\in[t,t+\Delta\,)|\,T^m(1,g)\ge t,\,X}}{\lim_{\Delta\to 0_+}\Pr{T^m(0,g)\in[t,t+\Delta\,)|\,T^m(0,g)\ge t,\,X}}\,.\label{eq:hrCausal}
\end{align}
These authors point out that the hazard functions in the numerator and the denominator of \eqref{eq:hrCausal} condition on different sets of subjects for $t > t_{(1)} = \min_{ij,g,z}T^m_{ij}(g,z)$. The causal contrast is thus some combination of treatment efficacy and a selection effect. Initially, randomized assignment of $Z$ and ``as-if" randomization of $G$ ensure that the units in each combination of trial arm and HbAS status are comparable on average. However, conditioning on the study subjects at risk of their first isolated malaria fever at time $t > t_{(1)}$, i.e., $T^m(z,g) \ge t$, has the potential to introduce selection bias. For instance, suppose that baseline susceptibility to malaria is heterogeneous even after conditioning on $X$. If $Z$ provides protection against developing an isolated malaria fever, we might find that the children with high susceptibility to malaria in the control arm are more likely to develop fevers than similar children in the treatment arm. Consequently, as time passes, the children at risk in the control arm will be less susceptible on average to fever than those in the treatment arm, absent treatment. The comparability of the subjects in each combination of $z,g \in \{0,1\}^2$ ensured by Assumption \ref{assump:random} at the beginning of the study is not guaranteed as the follow-up progresses. 

  When malaria-attributable fever is rare or when the follow-up time is short, the selection effect may be less consequential \citep{aalen2015does}. If $X$ does not sufficiently capture heterogenous susceptibility to isolated malaria fever, modeling subject-level frailty can further alleviate the selection bias of estimates of $\tau$ \citep{Wienke_2010}. Frailty can be introduced as a multiplicative random effect, 
  \begin{align}
    \lambda_{zg}^m(t\,|\,X) & = \lambda(t)W\exp\{\log\alpha + \log(1-\tau)z+\log(1-\nu)g + \boldsymbol{\beta}^TX\}\,,\notag \\
    \lambda_{zg}^{nm}(t\,|\,X) & = \lambda(t)W\exp\{\log\phi + \log(1-\eta)z + \boldsymbol{\beta}^TX\}\label{eq:frailty}
  \end{align}
where $W$ is a random, subject-level frailty shared by both malaria and non-malaria fever hazard functions. Proposition \ref{prop:PHidentify} extends immediately to the frailty model implied by \eqref{eq:frailty} where, in addition to $X$, the hazard ratio is conditional on $W$. Estimation can be carried out via the Expectation-Maximization algorithm \citep{klein1992} or penalized partial maximum likelihood methods \citep{therneau2003penalized}.

Alternatively, frailty models can be used to conduct a sensitivity analysis of Cox regression-based estimates of $\tau$ to selection bias \citep{Stensrud_2017}. In general, however, modeling frailty usually requires parametric assumptions about $W$, for example, that it comes from the family of power variance function distributions \citep{Wienke_2010}.

Although $\tau$ itself has a subtle and potentially awkward causal interpretation, simple functions of $\tau$ have been shown to have more natural causal interpretations. For example, $1/(2-\tau)$ is the {\it probabilistic index}, which is defined as the probability that $T^m(1,g)$ for one individual is longer than $T^m(0,g)$ for another individual with comparable baseline covariates (and frailty) \citep{De_Neve_2019}. 

Importantly, all the causal interpretations discussed in this section require the correct identification of the parameter $\tau$ established in Proposition \ref{prop:PHidentify}. 

\section{Technical Appendices}

\subsection{Proof of Proposition \ref{prop:identify}}\label{app:prop1}
In this Appendix we give a proof of the general identification result in Proposition \ref{prop:identify}.

\begin{proof}[Proof of Proposition \ref{prop:identify}]
  By consistency and Assumption \ref{assump:nointerfere} we have that 
\begin{equation}
  \Exp{f(Y)|\,X,Z=z,G=g}{P} = \Exp{f(Y(z,g))|\,X,Z=z,G=g}{P}
\end{equation}
  for all $z,g$. Assumption \ref{assump:random} ensures that $P_{X|Z,G}=P_{X}$ and $P_{Y(z,g)|Z,G}=P_{Y(z,g)}$. Hence,
  \begin{align}
    \Exp{\Exp{f\{Y(z,g)\}|\,X,Z=z,G=g}{P}}{X} & = \Exp{\Exp{f\{Y(z,g)\}|\,X,Z=z,G=g}{P}}{X|Z=z,G=g} \notag\\
    & = \Exp{f\{Y(z,g)\}|\,Z=z,G=g}{P} \notag\\
    & = \Exp{f\{Y(z,g)\}}{P} \notag\\
    & = \mu_{zg}(P)\,.
  \end{align}
  Applying Assumption \ref{assump:add}, the right hand side of \eqref{eq:identify} becomes
\begin{equation}\label{eq:ratio1} 
  1 - \frac{(\mu_{11}^m(P)-\mu_{10}^m(P)) + (\mu_{11}^{nm}(P)-\mu_{10}^{nm}(P))}{(\mu^m_{01}(P)-\mu^m_{00}(P)) + (\mu^{nm}_{01}(P)-\mu^{nm}_{00}(P))}\,.
\end{equation}
  Because $G$ is a valid Mendelian factor (Assumption \ref{assump:valid}(ii)) we have that $\mu^{nm}_{z1}(P)-\mu^{nm}_{z0}(P)=0$ for $z=0,1$, further simplifying \eqref{eq:ratio1} to 
  \begin{equation*}
    1 - \frac{(\mu_{11}^m(P)-\mu_{10}^m(P))}{(\mu^m_{01}(P)-\mu^m_{00}(P))}\,.
  \end{equation*}
  Finally, we have that 
  \begin{align*}
    1 - \frac{(\mu_{11}^m(P)-\mu_{10}^m(P))}{(\mu^m_{01}(P)-\mu^m_{00}(P))} & = 1 - \frac{-\nu\mu_{10}^m(P)}{-\nu\mu_{00}^m(P)}\quad\text{by Assumption \ref{assump:nointeract} and Definition \ref{defn:efficacy}(ii)}\\
      & = 1 - \frac{\mu_{10}^m(P)}{\mu_{00}^m(P)}\\
      & = \tau \qquad\qquad\qquad\;\;\, \text{by Definition \ref{defn:efficacy}(i)}\,.
  \end{align*}
  Assumption \ref{assump:valid}(i) ensures that the right hand side of the first equality is well-defined.
\end{proof}

\subsection{Proof of Proposition \ref{prop:bnd_estimator}}\label{app:prop3}
We give a short proof of the asymptotic unbiasedness of $\hat\tau_{bnd}$ and the asymptotical validity of $\text{CI}_{bnd,\alpha}$ as a $1-\alpha$ confidence interval for $\tau$.
\begin{proof}[Proof of Proposition \ref{prop:bnd_estimator}]
We give a proof for when the sample sizes are balanced across sites. When this is the case, we skip step 2 of Algorithm \ref{alg:est} and use $\hat\mu_0$ to construct $\hat\tau_0$. Treating $G$ as an element of $X$, the consistency of $\hat\tau_0$ for $s\tau + (1-s)\eta$ and its asymptotic linearity follows almost immediately from Theorem 1 of \cite{Rosenblum_2010}. When $\eta < \tau$, the theorem implies that $\hat\tau_0$ is consistent for some $\tau' < \tau$. Because $\tau \le 1$ by definition, the continuous mapping theorem implies that $\hat\tau_{bnd}$ is consistent for $\min[1,\max\{\tau',\tau\}]= \tau$. The lower bound of $\text{CI}_{bnd,\alpha}$ is constructed by taking the intersection of a lower one-sided confidence interval for $\tau$ with coverage $1-\alpha/2+\alpha_0$ and a lower one-sided confidence interval for $\tau'$ with coverage $1-\alpha_0$. Because $\tau'<\tau$, this second confidence interval is also asymptotically valid for $\tau$. If $\alpha_0$ is chosen a priori, then the intersection is an asymptotically valid $1-\alpha/2$ lower one-sided interval for $\tau$ \citep{neuwald1994detecting}. The upper confidence bound is constructed by taking the intersection of a $1-\alpha/2$ upper one-sided confidence interval for $\tau$ and $(-\infty,1]$. Since $\tau\le1$ the resulting interval is an asymptotically valid $1-\alpha/2$ upper confidence-interval for $\tau$. The intersection of the resultant upper and lower $1-\alpha/2$ one-sided confidence intervals yields an asymptotically valid two-sided confidence interval with coverage $1-\alpha$.
\end{proof}

\subsection{Proof of Proposition \ref{prop:PHidentify}}\label{app:prop2}
In this Appendix we provide a proof of the parameter identification result for the hazard ratio under the proportional hazards assumption.

\begin{proof}[Proof of Proposition \ref{prop:PHidentify}]
  Let $\Lambda_{zg}(t\,|\,X) = \int_{-\infty}^t \lambda_{zg}(t\,|\,X)\,dt$ be the cumulative hazard function for $Y$. Define $\Lambda_{zg}^k(t\,|\,X)$, $k=m,nm$ similarly. The survival function $S_{zg}(t\,|\,X = \Pr{T(z,g) > t\,|\,X}$ can be expressed as $S_{zg}(t\,|\,X) = \exp\{-\Lambda_{zg}(t\,|\,X)\}$ and the probability density of $Y$ as $f_{zg}(t\,|\,X) = \lambda_{zg}(t\,|\,X)S_{zg}(t\,|\,X)$. We can express $f_{zg}(t\,|\,X)$ as 
  \begin{align} 
    f_{zg}(t\,|\,X) & = \frac{\partial}{\partial\,t}F_{zg}(t\,|\,X) \notag\\
    & = \frac{\partial}{\partial\,t}\left\{1 - S_{zg}^{nm}(t\,|\,X)S_{zg}^m(t\,|\,X)\right\}\quad\text{by Assumption \ref{assump:indepFev}}\notag\\
    & = f_{zg}^m(t\,|\,X)S_{zg}^{nm}(t\,|\,X) + f_{zg}^{nm}(t\,|\,X)S_{zg}^m(t\,|\,X) \notag\\
    & = \{\lambda_{zg}^m(t\,|\,X) + \lambda_{zg}^{nm}(t\,|\,X)\}S_{zg}^{nm}(t\,|\,X)S_{zg}^m(t\,|\,X) \notag\\
    & = \{\lambda_{zg}^m(t\,|\,X) + \lambda_{zg}^{nm}(t\,|\,X)\}\times  \exp\left\{-\int_{-\infty}^t \lambda_{zg}^m(t\,|\,X) + \lambda_{zg}^{nm}(t\,|\,X)\,dt\right\}\,.\label{eq:hazY}
  \end{align}
  From \eqref{eq:hazY} and Assumptions \ref{assump:prop} and \ref{assump:shared}, we have that the hazard function for $Y$ is 
  \begin{align}
    \lambda_{zg}(t\,|\,X) & = \lambda_{zg}^m(t\,|\,X) + \lambda_{zg}^{nm}(t\,|\,X)\notag\\
    & = \lambda(t)\exp\{\boldsymbol{\beta}^TX\}\left(\kappa(1-\tau)^z(1-\nu)^g + (1-\eta)^z\right)\notag\\
    & = \lambda(t)\exp\{\alpha + \omega z + \gamma g + \lambda z\times g + \boldsymbol{\beta}^TX\}\,,\label{eq:factorialPH}
  \end{align}
  Proving the first part of the proposition. To prove the second part of the proposition we first observe that the last equality above in \eqref{eq:factorialPH} is simply a reparameterization -- when $z$ and $g$ are binary, $\kappa(1-\tau)^z(1-\nu)^g + \phi(1-\eta)^z$ and $\exp\{\alpha + \omega z + \gamma g + \lambda z\times g\}$ can each take four distinct values. The reparameterization yields a system of four equations,
  \begin{align}
    \exp\{\alpha\} & = \kappa + \phi\label{eq:eq00}\\
    \exp\{\alpha+\omega\} & = \kappa(1-\tau) + \phi(1-\eta) \label{eq:eq10}\\
    \exp\{\alpha+\gamma\} & = \kappa(1-\nu) \label{eq:eq01}\\
    \exp\{\alpha+\omega+\gamma+\lambda\} & = \kappa(1-\tau)(1-\nu) + \phi(1-\eta)\label{eq:eq11}\,.
  \end{align}
  Subtracting \eqref{eq:eq10} from \eqref{eq:eq11} and \eqref{eq:eq00} from \eqref{eq:eq01}, then dividing the former by the latter yields
  \begin{equation}
    1-\tau = \frac{\exp\{\omega + \gamma + \lambda\}-\exp\{\omega\}}{\exp\{\gamma\}-1}\label{eq:tau_identify}\,.
  \end{equation}
  Under Assumptions \ref{assump:nointerfere} and \ref{assump:random} we have that $\lambda_{zg}(t\,|\,X)$ can be identified by the observed data $\mathbf{O}$:
  \begin{align} 
    \lambda_{zg}(t\,|\,X) & = \lim_{\Delta \to 0_+}\Pr{T(z,g)\in[t,t+\Delta)|\,T(z,g)\ge t,\, X}/\Delta \notag\\
    & = \lim_{\Delta \to 0_+}\Pr{T\in[t,t+\Delta)|\,T\ge t,\, X,\,Z=z,\,G=g}/\Delta \label{eq:PHid_fromData}\,.
  \end{align}
  Applying \eqref{eq:factorialPH}, we have that $\alpha,\omega,\gamma,\lambda$ are identifiable and thus $\tau$ can be identified by the observed data using \eqref{eq:tau_identify}.
\end{proof}


\subsection{Proof (Sketch) of Proposition \ref{prop:tmle}}\label{app:tmle}

In this section we sketch a proof of Proposition \ref{prop:tmle} by showing that $\mu_{zg}(P_n)$ is consistent for $\mu_{zg}(P)$ and asymptotically linear, that is,
\begin{equation}
  \mu_{zg}(P_n) - \mu_{zg}(P) = (\mathbb{P}_n - \mathbb{P})\varphi^*_{zg}(P_\infty) + o_P\left(1/\sqrt{N}\right)\,.\label{eq:ALresult}
\end{equation}  

for all $z,g$. Theorem 1 of \cite{Rosenblum_2010} checks a number of technical conditions to verify that Theorem 1 of \cite{van_der_Laan_2006} can be applied. A particularly important condition is that the $\mu_{zg}(P)$ are linear. This does not hold in the setting when the prevalence of $G$ in each site is not known a priori. Hence, a straightforward application of the theorem is not appropriate. However, if certain weaker conditions hold, then the important parts of Theorem 1 of \cite{van_der_Laan_2006} still apply. We first show that we can write
\begin{equation}
  \mu_{zg}(P_n) - \mu_{zg}(P) = (\mathbb{P}_n - \mathbb{P})\varphi^*_{zg}(P_n) \,.\label{eq:bigO_al}
\end{equation}  

If $\varphi^*_{zg}(P_n)$ is Donsker then by the second part of Theorem 1 of \cite{van_der_Laan_2006}, we have that $\mu_{zg}(P_n)$ is $\sqrt{n}$-consistent for $\mu_{zg}(P)$. Then, if $\mathbb{P}\{\varphi^*_{zg}{P_n} - \varphi^*_{zg}(P_\infty)\}^2 \to 0$ in probability as $n\to\infty$ the third part of Theorem 1 of \cite{van_der_Laan_2006} implies that $\mu_{zg}(P_n)$ is asymptotically linear with form \eqref{eq:ALresult}. Consistency of $\mu_{zg}(P_n)$ and Assumptions \ref{assump:add} - \ref{assump:valid} imply that the plugin estimator $\hat\tau$ is consistent for $\tau$. Finally, we apply the delta method to derive the asymptotic variance $\sigma^2$ of $\sqrt{n}(\hat\tau-\tau)$. It follows from the last part of Theorem 1 of \cite{van_der_Laan_2006} that $\sigma^2$ achieves the semiparametric efficiency bound if the working model for the conditional expectation of $f(Y)$ is correctly specified. In what follows, when taking expectations of functionals of $P_n$ we treat these functionals as fixed. \\

\noindent{\it Verifying  \eqref{eq:bigO_al}:}

Given the choice of the terms included in the linear part of the GLM estimate $\hat\mu_1$ in Algorithm \ref{alg:est}, the score equations that $\hat\mu_1$ solve imply that
\begin{equation}
	\mathbb{P}_n\left\{ \frac{\mathbbm{1}(Z=z)\mathbbm{1}(G=g)}{p(Z=z)p_{j,n}(G=g)}\left(f(Y) - \hat\mu_1(z,g,X)\right)\right\} = 0\;\text{ for all }\; z,g \in \{0,1\}^2\,.\label{eq:asymlin}
\end{equation}
The choice of the weighting $\texttt{w}=n/I_j$  when estimating $\hat\mu_1$ is required since $\mathbb{P}_n$ does not equally weight observations unless the size of each site is the same. Also note, that by definition, $\mathbb{P}_n\hat\mu_1(z,g,X) =\mu_{zg}(P_n)$. Taken together, we have that $\mathbb{P}_n\varphi_{zg}(P_n) = 0$ for all $z,g\in\{0,1\}^2$. By iterated expectations, we have that
\begin{equation}
	\mathbb{P}\varphi^*_{zg}(P_n) = \mathbb{P}\varphi_{zg}(P_n) - \mathbb{P}\left\{\Exp{\varphi_{zg}(P_n)\,|\,G}{P}\right\} =  \mathbb{P}\varphi_{zg}(P_n) - \mathbb{P}\varphi_{zg}(P_n) = 0\,.
\end{equation}

All that remains to be shown is that $-\mathbb{P}_n\left\{\Exp{\varphi_{zg}(P_n)\,|\,G}{P}\right\} = \mu_{zg}(P_n) - \mu_{zg}(P)$. We start by evaluating $\Exp{\varphi_{zg}(P_n)\,|\,G}{P}$:
\begin{align}
	\Exp{\varphi_{zg}(P_n)\,|\,G}{P} & = \Exp{\frac{\mathbbm{1}(Z=z)\mathbbm{1}(G=g)\left\{f(Y) - \hat\mu_1(z,g,X)\right\}}{p_{j,n}(G=g)p(Z=z)}\,\Big{|}\,G}{P}\notag\\
	& \quad+ \Exp{\hat\mu_1(z,g,X)- \mu_{zg}(P_n)\,|\,G}{P} \notag\\
	& = \frac{\mathbbm{1}(G=g)}{p_{j,n}(G=g)}\left(\mu_{zg}(P) - \Exp{ \hat\mu_1(z,g,X)\,|\,G}{P}\right) \notag\\
	& \quad+ \Exp{\hat\mu_1(z,g,X)\,|\,G}{P} - \mu_{zg}(P_n)\,.
\end{align}
The second equality follows from an application of iterated expectations, Assumption \ref{assump:random}, and the fact that $\mu_{zg}(P_n)$ and $\hat\mu_1$ are treated as fixed when taking expectations. Taking the empirical expectation $\mathbb{P}_n$ we get
\begin{align}
	\mathbb{P}_n\left\{\Exp{\varphi_{zg}(P_n)\,|\,G}{P} \right\} & = \mathbb{P}_n\left\{\frac{\mathbbm{1}(G=g)}{p_{j,n}(G=g)}\left(\mu_{zg}(P) - \Exp{ \hat\mu_1(z,g,X)\,|\,G}{P}\right)\right\} \notag\\
	& \quad+ \mathbb{P}_n\left\{\Exp{\hat\mu_1(z,g,X)\,|\,G}{P} - \mu_{zg}(P_n)\right\}\notag\\
	& = \mu_{zg}(P) -  \Exp{ \hat\mu_1(z,g,X)\,|\,G=g}{P}  \notag\\
	&\quad+ \mathbb{P}_n\left\{\Exp{\hat\mu_1(z,g,X)\,|\,G}{P}\right\} - \mu_{zg}(P_n) \notag\\
	& =  \mu_{zg}(P) -  \Exp{ \hat\mu_1(z,g,X)}{P} +   \Exp{ \hat\mu_1(z,g,X)}{P} - \mu_{zg}(P_n)\notag\\
	& = -(\mu_{zg}(P_n) - \mu_{zg}(P))\,.
\end{align}
The second to last equality comes from the fact that $X\indep G$ by Assumption \ref{assump:random}. This is our desired result and \eqref{eq:bigO_al} now follows.\\

\noindent{\it Checking that $\varphi^*_{zg}(P_n)$ is Donsker:}

This is analogous to condition (iv) in the proof of Theorem 1 in \cite{Rosenblum_2010}. We make the same assumptions of boundedness on $\boldsymbol{\beta}$ and on the terms in the linear part of the generalized linear models in steps 1 and 2 of Algorithm \ref{alg:est}. Under the additional assumption that $0<\delta\le p_j(G=g)$ for all $g=0,1$ and $j = 1,\dots,J$, a nearly identical verification that $\varphi^*_{zg}(P_n)$ is Donsker follows. Hence, $\mu_{zg}(P_n)$ are $\sqrt{n}$-consistent for all $z,g\in\{0,1\}1^2$.\\

\noindent{\it Verifying that $\mathbb{P}\{\varphi^*_{zg}(P_n) - \varphi^*_{zg}(P_\infty)\}^2 \stackrel{P}{\to} 0$:}

The steps to verify that $\mathbb{P}\{\varphi^*_{zg}(P_n) - \varphi^*_{zg}(P_\infty)\}^2 \to 0$ in probability are very similar to the verification of condition (v) in \cite{Rosenblum_2010}. There are a few extra steps required to deal with the fact that we are estimating $p_j(G=g)$ with $p_{j,n}(G=g)$. We first note that we can write
\begin{align}
	\mathbb{P}\{\varphi^*_{zg}{P_n} - \varphi^*_{zg}(P_\infty)\}^2  & = \mathbb{P}[\{\varphi_{zg}(P_n) - \varphi_{zg}(P_\infty)\} - \Exp{ \varphi_{zg}(P_\infty)-\varphi_{zg}(P_n)\,|\,G}{P} ]^2 \notag\\
	& \le  2\mathbb{P}\{\varphi_{zg}(P_n) - \varphi_{zg}(P_\infty)\}^2 + 2\mathbb{P}\{\Exp{ \varphi_{zg}(P_\infty)-\varphi_{zg}(P_n)\,|\,G}{P} \}^2 \notag\\
	& \le  4\mathbb{P}\{\varphi_{zg}(P_n) - \varphi_{zg}(P_\infty)\}^2\,.\label{eq:bnds1}
\end{align}
The last line comes from applying Jenson's inequality to the inner expectation of the second term followed by an application of iterated expectations. We need to show that $\mathbb{P}\{\varphi_{zg}(P_n) - \varphi_{zg}(P_\infty)\}^2$ converges to $0$ in probability. In what follows we distinguish the working model using the fitted parameters $\boldsymbol{\beta}_n$ and the model using the unique maximizer of the expected log-likelihood $\boldsymbol{\beta}$ as $\hat\mu_{1,n}$ and $\hat\mu_{1,\infty}$, respectively. For brevity, we also let $\mathbbm{1}_{zg} = \mathbbm{1}(Z=z)\mathbbm{1}(G=g)$, $p(z) = p(Z=z)$, $p_{j}(g) = p_j(G=g)$, and $p_{j,n}(g) = p_{j,n}(G=g)$. We can write
\begin{align}
	&\mathbb{P}\{\varphi_{zg}(P_n) - \varphi_{zg}(P_\infty)\}^2 \notag\\
	& \qquad = \mathbb{P}\left\{ \frac{\mathbbm{1}_{zg}} {p(z)p_{j,n}(g)}\{f(Y) - \hat\mu_{1,n}(z,g,X)\} + \hat\mu_{1,n}(z,g,X) - \mu_{zg}(P_n)  \right.\notag\\
	& \qquad\qquad\left. -  \frac{\mathbbm{1}_{zg}} {p(z)p_{j}(g)}\{f(Y) - \hat\mu_{1,\infty}(z,g,X)\} - \hat\mu_{1,\infty}(z,g,X) + \mu_{zg}(P) \right\}^2 \notag\\
	&\qquad \le \frac{4}{p(z)^2}\mathbb{P}\left\{f(Y)^2\right\}\left(\frac{1}{p_{j,n}(g)}-\frac{1}{p_j(g)} \right)^2 + \frac{4}{p(z)^2}\mathbb{P}\left\{\frac{\hat\mu_{1,\infty}(z,g,X)}{p_{j}(g)}-\frac{\hat\mu_{1,n}(z,g,X)}{p_{j,n}(g)}\right\}^2\notag\\
	& \qquad\qquad + 4\mathbb{P}\left\{\hat\mu_{1,n}(z,g,X) - \hat\mu_{1,\infty}(z,g,X)\right\}^2 + 4\mathbb{P}\left\{\mu_{zg}(P)-\mu_{zg}(P_n)\right\}^2 \notag\\
	& \qquad \le C_1\left(\frac{1}{p_{j,n}(g)}-\frac{1}{p_j(g)} \right)^2 + \frac{4}{p(z)^2}\mathbb{P}\left\{\frac{\hat\mu_{1,\infty}(z,g,X)}{p_{j}(g)}-\frac{\hat\mu_{1,n}(z,g,X)}{p_{j,n}(g)}\right\}^2\notag\\
	& \qquad\qquad + C_2\pnorm{\boldsymbol{\beta}_0-\boldsymbol{\beta}}^2 + 4\{\mu_{zg}(P)-\mu_{zg}(P_n)\}^2 \,.\label{eq:bnds2}
\end{align}
The first equality follows immediately from definitions. The first {\it inequality} follows from rearranging terms and noting that $(x_1+\dots+x_k)^2 \le k(x_1^2+\dots+x_k^2)$ by Jensen's inequality. The first term after the second inequality comes from the boundedness of $Y$. Since $p_{j,n}(g) \stackrel{P}{\to} p_j(g)$ and we have assumed that $p_j(g)$ are bounded away from zero, the continuous mapping theorem implies that this term converges to $0$ in probability. The third term comes from the fact that our working model has uniformly bounded first derivatives. This is due to the boundedness assumptions on $X$ and the terms of the linear part of our working model. The conditions given in Proposition \ref{prop:tmle},  are sufficient for $\boldsymbol{\beta}_n$ to converge to $\boldsymbol{\beta}$ in probability (\cite{rosenblum2009using}, Appendix D). Consequently, the third term also converges to $0$ in probability. The fourth term disappears in probability since we proved earlier that $\mu_{zg}(P_n)$ is consistent for $\mu_{zg}(P)$. All that is left to handle is the second term. A little bit of rearranging gives us
\begin{align}
	 &\frac{4}{p(z)^2}\mathbb{P}\left\{\frac{\hat\mu_{1,\infty}(z,g,X)}{p_{j}(g)}-\frac{\hat\mu_{1,n}(z,g,X)}{p_{j,n}(g)}\right\}^2 \notag\\
	 &\qquad =   \frac{4}{p(z)^2}\mathbb{P}\left\{ \frac{\hat\mu_{1,\infty}(z,g,X)-\hat\mu_{1,n}(z,g,X)}{p_{j}(g)} +  \hat\mu_{1,n}(z,g,X)\left(\frac{1}{p_{j,n}(g)}-\frac{1}{p_j(g)}\right)     \right\}^2\notag\\
	 & \qquad \le \frac{8}{\{p(z)p_j(g)\}^2}\mathbb{P}\left\{\hat\mu_{1,\infty}(z,g,X)-\hat\mu_{1,n}(z,g,X)\right\}^2 \notag\\
	 &\qquad\qquad +  \frac{8}{p(z)^2}\mathbb{P}\left\{\hat\mu_{1,n}(z,g,X)^2\right\}\left(\frac{1}{p_{j,n}(g)}-\frac{1}{p_j(g)} \right)^2 \notag\\
	 & \qquad = C_3\pnorm{\boldsymbol{\beta}_0-\boldsymbol{\beta}}^2 + C_4\left(\frac{1}{p_{j,n}(g)}-\frac{1}{p_j(g)} \right)^2\,.\label{eq:bnds3}
\end{align}
The first inequality follows again from Jenson's inequality. The first term in the last line follows from the same arguments made earlier and the second term follows from the fact that $\mathbb{P}\left\{\hat\mu_{1,n}(z,g,X)^2\right\}$ is bounded. As we've already demonstrated, these two terms vanish in probability. Combining \eqref{eq:bnds1}, \eqref{eq:bnds2}, and \eqref{eq:bnds3} gives us that $\mathbb{P}\{\varphi^*_{zg}(P_n) - \varphi^*_{zg}(P_\infty)\}^2 \stackrel{P}{\to} 0$. \eqref{eq:asymlin} follows immediately and an application of Proposition \ref{prop:identify} implies that $\hat\tau$ is also asymptotically linear. All that remains is to compute its asymptotic variance.\\

\noindent{\it Asymptotic Variance:}

To compute the asymptotic variance of $\hat\tau$ we begin with a taylor expansion around $\tau$:
\begin{align} 
  \hat\tau & = 1 - \frac{\mu_1(P_n)}{\mu_0(P_n)} \notag \\
  & = 1 - \frac{\mu_1(P)}{\mu_0(P)} - \frac{1}{\mu_0(P)}\{\mu_{1}(P_n) - \mu_{1}(P)\} \notag \\
  & \qquad\quad + \frac{\mu_1(P)}{\mu_0(P)^2}\{\mu_{0}(P_n) - \mu_{0}(P)\} + o_p\left(1/\sqrt{n}\right)\notag \\
  & = \tau - \frac{1}{\mu_0(P)}(\mathbb{P}_n - \mathbb{P})\varphi^*_{1}(P_\infty)(O) + \frac{\mu_1(P)}{\mu_0(P)^2}(\mathbb{P}_n - \mathbb{P})\varphi^*_{0}(P_\infty)(O) + o_p(1/\sqrt{n})\label{eq:taylor}\,.
\end{align}
The $o_p(1/\sqrt{n})$ comes from the second order remainder term of the expansion. Using the expansion in \eqref{eq:taylor} along with the fact that $\varphi^*_{z}(P_\infty)(O) - \mathbb{P}\varphi^*_{z}(P_\infty)(O)$ are mean zero i.i.d. random variables we can write the scaled asymptotic variance of $\hat\tau$, $n\cdot\var(\hat\tau)$, as 
\begin{equation}
  \Exp{\frac{\mu_1(P)}{\mu_0(P)^2}(\varphi^*_{0}\{P_\infty)(O) - \mathbb{P}\varphi^*_{0}(P_\infty)(O)\} - \frac{1}{\mu_0(P)}\{\varphi^*_{1}(P_\infty)(O) - \mathbb{P}\varphi^*_{1}(P_\infty)(O)\} }{P}^2\,.
\end{equation} 

\subsection{Additional Simulation Study Details for TMLE-based Estimators}\label{app:simdetail}

Below are the simulation settings for the individual-specific vaccine and protective efficacies ($\tau_i$ and $\nu_i$), the baseline covariates and how they are transformed when they enter the condition mean of the distribution of fever counts ($X_i$, $\tilde X^m_i$, and $\tilde X^{nm}_i$), and unobserved heterogeneity in the distribution of fever counts between individuals ($\epsilon^m_i$ and $\epsilon^{nm}_i$).\\

\noindent{\it Generative Distributions:}
\begin{align}
  &(1-\nu_i) \sim (1-\nu)\exp\{-0.05^2/2\}\cdot \texttt{Lognormal}(0,0.05^2)\notag\\
  &(1-\tau_i) \sim (1-\tau)\exp\{-0.05^2/2\}\cdot \texttt{Lognormal}(0,0.05^2)\,,\notag\\
  &\tilde X^m_i \sim \exp\{0.05X_i - 0.05^2/2\}\,,\notag\\
  &\tilde X^{nm}_i \sim \exp\{0.075X_i - 0.075^2/2\}\,,\notag\\
  &X_i \sim \texttt{Normal}(0,1)\,,\notag\\
  &\epsilon^m_i \sim \exp\{-0.05^2/2\}\cdot\texttt{Lognormal}(0,0.05^2)\,,\notag\\
  &\epsilon^{nm}_i \sim \exp\{-0.05^2/2\}\cdot\texttt{Lognormal}(0,0.05^2)\,.\notag
\end{align}
We assume that the spillover efficacy $\eta$ is equal to zero. Based on empirical evidence and the negative dependence between malaria-attributable fevers and non-malaria fevers as defined in \eqref{eq:Ydecomp}, we simulate the number of malaria-attributable and non-malaria fevers from negatively dependent negative binomial distributions \citep{Olotu_2013}. We use a overdispersion parameter of $r = 10$\\

\noindent{\it Count of malaria-attributable Fevers per child-year:}

On the margin, the count of malaria-attributable fevers per child-year follows a negative binomial distribution with mean $\mu_{m,i}$ and variance $\sigma_{m,i}^2$,
\begin{equation}
  \1^T Y^m_i(z,g) \sim \texttt{NB}(\mu_{m,i},\sigma_{m,i}^2)
\end{equation}
where $\mu_{m,i} = \kappa\cdot(1-\nu_i)^g\cdot(1-\tau_i)^z\cdot\tilde X^m_i\cdot\epsilon^m_i$ and $\sigma_{m,i}^2 = \mu_{m,i}^2/r+\mu_{m,i}$.\\

\noindent{\it Count of non-malaria fevers per child-year:}

On the margin, the count of non-malaria fevers per child-year follows a negative binomial distribution with mean $\mu_{nm,i}$ and variance $\sigma_{nm,i}^2$,
\begin{equation}
  \1^T Y^{nm}_i(z,g) \sim \texttt{NB}(\mu_{nm,i},\sigma_{nm,i}^2)
\end{equation}
where $\mu_{nm,i} = \phi\cdot\tilde X^{nm}_i\cdot\epsilon^{nm}_i$ and $\sigma_{nm,i}^2 = \mu_{nm,i}^2/r+\mu_{nm,i}$.\\

\noindent{\it Count of fevers of any-cause per child-year:}

The joint distribution of the counts of non-malaria and malaria-attributable fevers per child-year are simulated using a Gaussian copula with a negative dependence parameter $\rho = -0.1$ \citep{genest2007primer}.\\

\noindent{\it Specificity and Mendelian Gene Prevalence Settings:}

We calibrated $\kappa$ and $\phi$ to achieve different levels of case specificity ($s = $ 0.5 and 0.8). For example, for a case specificity of 0.8 we set $\kappa$ and $\phi$ so that the expected number of malaria-attributable fevers is 80\% of the expected number of total fevers of any-cause. The prevalence of the Mendelian gene was set to 20\% based on existing estimates of HbAS prevalence in sub-Saharan Africa \citep{TER_KUILE_2003, elguero2015malaria}. \\

\noindent{\it Initial and Updated Working Model Estimates:}

For both the MFD and naive estimators, we used the \texttt{R} function \texttt{glm.fit} from the package \texttt{stats} to fit the initial estimator $\hat\mu_0$ in step 1 of Algorithm \ref{alg:est} using Poisson regression. The regression included main terms for $Z$, $G$, and $X$ as well as all interactions.

\bibliographystyle{chicago}
\bibliography{./biblio} 


\end{document}